\documentclass[num-refs]{wiley-networks}

\usepackage{amssymb}
\usepackage{mathrsfs}
\usepackage{algorithm}
\usepackage{algorithmic}
\usepackage{subcaption}
\usepackage{booktabs}
\usetikzlibrary{positioning, arrows.meta, calc}
\usepackage{pgfplots}
\pgfplotsset{compat=1.16}

\providecommand{\texorpdfstring}[2]{#1}

\papertype{Preprint}

\title{Betweenness-Central Nodes Under Uncertainty: An Absorbing Markov Chain Approach}

\author[1]{Wencheng Bao}
\author[1]{Eleftheria Kontou}
\author[2]{Chrysafis Vogiatzis}

\affil[1]{Department of Civil \& Environmental Engineering, University of Illinois Urbana-Champaign, Urbana, Illinois, 61801, USA}
\affil[2]{Department of Industrial \& Enterprise Systems Engineering, University of Illinois Urbana-Champaign, Urbana, Illinois, 61801, USA}

\corraddress{Chrysafis Vogiatzis, Department of Industrial \& Enterprise Systems Engineering, University of Illinois Urbana-Champaign, Urbana, Illinois, 61801, USA}
\corremail{chrys@illinois.edu}

\fundinginfo{National Science Foundation, Grant/Award Number: 2237881}

\runningauthor{Bao, Kontou, and Vogiatzis}

\begin{document}

\maketitle

\begin{abstract}
We propose a betweenness centrality measure and algorithms for stochastic networks, where edges can fail and weights vary across realizations, making the most central node random. Our approach models the sequence of reported central nodes as an absorbing Markov chain and measures node importance by the share of pre-absorption time spent at each node. This produces a way to study centrality under uncertainty, which can then be estimated with Monte Carlo simulation. We also analyze robustness when the transition kernel is only approximately known, using row-wise perturbations to assess sensitivity and potential ranking changes. The framework further admits extensions to weighted rewards and restricted candidate sets without altering the Markov chain formulation. Experiments on Erd\H{o}s--R\'enyi, Watts--Strogatz, and Les Mis\'erables networks with stochastic edges show that the method identifies a small set of dominant nodes, reveals stable versus sensitive rankings under perturbations, and supports reward-based and structure-constrained variants. 

\end{abstract}

\section{Introduction}
Centrality has been a prominent measure of importance in transportation, communication, and infrastructure networks \cite{cadini_zio_petrescu_2009, derrible_2012}. Among the many centrality metrics proposed over the years, betweenness centrality (BC) stands out as a natural fit in applications involving shortest paths and routing \cite{li2020identification, morelli_cunha_2021, sun_pei_hao_wang_zuo_wong_2018, chakrabarti_kushari_mazumder_2022}. Mathematically, the BC of node $n \in N$ can be calculated as \cite{wasserman_faust_1994}:
\begin{equation}
\mathrm{B}(n)=\sum_{s \neq n \neq t \in N} \frac{\sigma_{s, t}(n)}{\sigma_{s,t}}
\end{equation}
where $\sigma_{s,t}$ is the total number of shortest paths (also referred to as geodesics) between nodes $s$ and $t$, and $\sigma_{s, t}(n)$ is the total number of paths between the pair that pass through node $n$. 

Centrality has traditionally focused on deterministic and static settings, with some recent developments in stochastic and dynamic graphs (e.g., \cite{avella2018centrality}). In addition, nodal centrality metrics have been studied more than their group counterparts \cite{camur2024survey}. However, in many modern applications, the network observed is neither deterministic nor static: edges can fail (e.g., flooded streets in transportation networks), capacities and other edge characteristics can change (e.g., lane closures in a transportation network), and travel times and edge costs can be random and correlated in groups (e.g., travel times in a transportation network near a car crash). Under these circumstances, the node of highest betweenness centrality is itself random and can vary substantially across different realizations of the network.

A common approach to introducing uncertainty in BC calculations is to redefine the definition of the shortest path in stochastic networks, typically via randomized or probabilistic shortest path models. We then aggregate over these paths to compute expected shortest path BC or probabilistic BC variants \cite{kivimaki2016two, saha_brokkelkamp_velaj_khan_bonchi_2021}. The key idea is to identify candidate (in the form of expected or most probable) shortest paths between source and target pairs and measure importance by how frequently a node lies on these paths. However, this shortest path-centric approach becomes difficult once additional structural requirements are imposed, such as restricting paths to a specific subgraph or region, or enforcing other node or path level constraints \cite{parchas_gullo_papadias_bonchi_2015, horvath_kis_2017}.


Another widely used tool for analyzing importance is centrality based on random walk \cite{newman_2005}, including stationary visit measures such as PageRank \cite{brin_page_1998}, and other random walk counterparts of BC like current flow variants \cite{witt_2005, stolarksy_doyle_snell_1987}. 
These include walk- or flow-based BC within settings such as multilayer networks, dense graphs, or graphon limits, and higher order structures \cite{bottcher2021classical, petit_lambiotte_carletti_2021, schaub_benson_horn_lippner_jadbabaie_2020, benzi_klymko_2015}. Another novel approach is based on absorbing Markov chains (AMC): this moves away from the geodesic limit toward a limit where the path length is unrestricted for walkers \cite{gurfinkel_rikvold_2020}.


However, centrality metrics that are based on random walkers quantify where a walker tends to be (or how flow spreads) under exogenously specified Markov dynamics. Consequently, high stationary visitation probability for a node, which may be heavily influenced by nodal degree or by regularization, need not align with being a node of maximum BC, so walk-based rankings can drift toward frequently visited nodes rather than structurally mediating ones. 

\begin{figure}[htbp]
    \centering
    \includegraphics[width=0.6\linewidth]{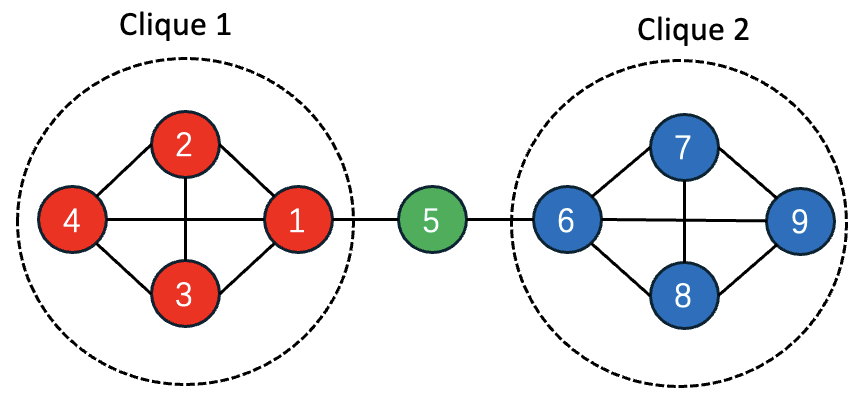}
    \caption{A deterministic graph consisting of two $4$-cliques connected by a length-two bridge $1$--$5$--$6$. Nodes $\{1,2,3,4\}$ form Clique~1 and nodes $\{6,7,8,9\}$ form Clique~2; node $5$ is the unique intermediary connecting the two cliques.}
    \label{fig:clique_example}
\end{figure}

In Figure~\ref{fig:clique_example}, the provided deterministic graph consists of two dense cliques connected only through the path $1 \to 5 \to 6$. Under the classical notion of BC, node $5$ is the unique maximizer since every shortest path with endpoints in different cliques has to pass through $5$. However, centrality metrics based on random walks measure long run visitation under a chosen Markov rule, and for an undirected simple random walk the stationary visit probability is proportional to node degree. Since $\deg(5)=2$ while $\deg(1)=\deg(6)=4$ (with all other clique nodes having degree $3$), node $5$ has stationary visit probability $2/28$, which is smaller than $4/28$ for nodes $1$ and $6$, so walk-based rankings favor nodes $1$ and $6$ over $5$.

Considering our context of a stochastic network, edge availability and edge weights may vary across realizations, leading to graph realizations that are disconnected. This is why centrality metrics based on random walk often rely on restarts or teleportation to remain well-defined. 
In a stochastic version of Figure~\ref{fig:clique_example}, the bridge edges $(1,5)$ and $(5,6)$ may be unavailable. When a bridge edge is absent, the graph becomes disconnected, and consequently paths connecting the two cliques disappear, meaning that the node of maximum centrality will vary sharply across realizations. Random walk centrality metrics then require a modeling choice. Without restarts, the walk is limited to its starting component, and no single global long-run ranking is meaningful. In the presence of restarts or teleportation, the realized connectivity and topology of the network instance is bypassed, which can blur any fragmentation effects \cite{lambiotte_rosvall_2012, masuda_porter_lambiotte_2017}. Moreover, computing walk- or flow-based scores over many realizations can be computationally expensive due to repeated stationary distribution calculations or Laplacian-type solves required \cite{cloninger_gal_mishne_oslandsbotn_robertson_wan_wang_2024}.

These observations motivate a complementary question: how does the reported central node itself evolve as the network is repeatedly realized? As an example, consider a monitoring resource (e.g., a sensor, or an inspection crew), stationed at the currently most central node. Each period, the network re-realizes (edges fail or recover, travel times change); then, the resource observes the connected component it currently occupies and relocates to the betweenness maximizer of that component. If the component has become too small to support an informative ranking, the episode ends. The long-run fraction of time the resource spends at each node is the quantity our framework computes and returns, and the anchoring of each realization at the previous center reflects the physical constraint that the resource can only observe and act within the part of the network it currently occupies.


\subsection{Our contributions}

We propose a unified framework for studying betweenness centrality in stochastic networks. Instead of recomputing a ranking from scratch for each realized graph, we model how the reported central node changes across these realizations. This leads to an absorbing Markov chain on the node set together with a terminal state. We call this \textbf{absorbing-frequency centrality} (AFC), which measures the normalized fraction of pre-absorption time spent at each node. 

We define local centers (component-wise) relative to a current anchor node, addressing modeling issues that arise in possibly disconnected graphs. We introduce a component-size threshold that triggers absorption when the component including the anchor node becomes too small to support informative ranking. We also clarify why the absorbing-chain representation tracks a single reported center. In an anchor-free special case, we show that AFC admits a simple mixture representation when the post-initial center law is stationary, conditional on survival. We further characterize when this simplification breaks down under heterogeneous transition dynamics.

We further extend the framework to row-wise perturbation analysis, reward-based AFC, and Top-$k$ or structure-constrained selection, all within the same absorbing-chain representation. A row-wise Monte Carlo procedure with finite-sample guarantees makes the method implementable, and experiments on the Erd\H{o}s--R\'enyi, Watts--Strogatz, and Les Mis\'erables networks show that AFC identifies dominant nodes, reveals ranking sensitivity, and supports reward-aware and constrained analyses using a single estimated chain. 

\subsection{Paper roadmap}

First, we introduce the proposed framework through a small worked example in Section~\ref{sec:glance}. The example is designed to separate it from four established alternatives:
\begin{enumerate}
    \item classical betweenness on the base graph,
    \item realization-averaged betweenness,
    \item random-walk betweenness \cite{newman_2005}, and
    \item and absorbing random-walk centrality \cite{gurfinkel_rikvold_2020}.
\end{enumerate}

The rest of the paper develops the theory in the order needed for the main results. After introducing the stochastic network model, global/component-wise betweenness centers, and local Top-$k$ sets, we construct a node-valued AMC on reported local centers and establish its well-posedness. In Subsection~\ref{subsec:single-node-uniqueness}, Proposition~\ref{prop:single-node-uniqueness} shows that unresolved candidate correspondences need not induce a unique node-valued row, Corollary~\ref{cor:tie-break-unique-row} gives the canonical row under deterministic tie-breaking, and Proposition~\ref{prop:set-state-not-unique} shows that propagating the whole candidate set does not remove this ambiguity. In Subsection~\ref{subsec:wellposed}, Proposition~\ref{prop:order-invariance} and Lemma~\ref{lem:row-normalization} show that the compressed row is determined by the pushforward realized-graph law and absorption rule and is a valid probability row. Subsection~\ref{subsec:uniform-step} then gives the trajectory interpretation of AFC: Proposition~\ref{prop:uniform-step} identifies AFC with a length-biased uniformly sampled pre-absorption step, and Corollary~\ref{cor:survival-decomp} yields the survival-weighted decomposition used throughout the sequel.

Building on this representation, Section~\ref{subsec:stochastic} studies arbitrary initial distributions in the connected anchor-free regime, where Theorem~\ref{thm:mixture}, Corollary~\ref{cor:geometric}, and Proposition~\ref{prop:linear-algebra} provide the mixture benchmark and its closed-form fundamental matrix representation. Section~\ref{subsec:robust-perturb-b} then analyzes row-wise perturbations, showing that the post-initial stationarity required by Theorem~\ref{thm:mixture} typically fails, while Remark~\ref{rem:pm-transience}, \eqref{eq:pm-survival-average}, the first-order sensitivity expansion, and the visit-gap bounds still yield perturbation and ranking-stability diagnostics. Section~\ref{subsec:multi-node-reward} next extends the same AMC through rewards: Proposition~\ref{prop:reward-afc-general} reduces simulator-level rewards to the same fundamental-matrix formula, leading to valued local Top-$k$ summaries, pool-restricted rewards, and node- and transition-based reward-AFC without enlarging the state space. Section~\ref{subsec:set-topk-phi-QN} complements this by modifying the selector itself, leading to the constrained kernels $P^W$ and $P^{W,\mathrm{fb}}$, the feasibility probabilities \eqref{eq:settopk-xi}, and the pool-mass summaries \eqref{eq:settopk-bPsi}. Section~\ref{subsec:construct-amc} then turns the theory into an implementable procedure via Algorithms~\ref{alg:sample-next} and~\ref{alg:construct-amc}, proves consistency of the row-wise Monte Carlo estimator, and gives finite-sample guarantees for estimating both the one-shot law $p$ and the full kernel $P$, together with error propagation to AFC. Finally, Section~\ref{sec:numerical-experiments} illustrates the full pipeline on ER and WS random graphs and on the Les Mis\'erables network through baseline AFC, robust perturbation analysis, reward-aware summaries, and structure-constrained experiments.

\section{Absorbing-frequency centrality at a glance}
\label{sec:glance}

In this section, we introduce the proposed measure informally and through a worked example. Let $G(V,E)$ be a base topology. A stochastic network generates random working graphs $H$ by failing edges and/or perturbing edge weights. For each realized $H$, classical betweenness is well defined on every connected component, and a deterministic tie-breaker selects a unique reported center in the component containing a chosen anchor. We track the sequence of reported centers as the network is repeatedly resampled. That is, from the current center $i$, draw a new realization anchored at $i$; if the component of $i$ is smaller than a threshold $k_{\min}$, or an exogenous stopping event occurs, move to an absorbing state $\perp$; otherwise move to the tie-broken betweenness maximizer of that component. The result is an absorbing Markov chain with transition matrix of the form \eqref{eq:amc}, and \textbf{absorbing-frequency centrality} (AFC) is the normalized expected pre-absorption occupancy $b(s)=sN/sN\mathbf 1$ with $N=(I-Q)^{-1}$, formally defined in \eqref{eq:afc-def}. High AFC means a node is reported as a local center often \emph{and} tends to remain so under the chain's own dynamics.

\subsection{A worked example: the asymmetric crossroads}
\label{subsec:crossroads}

Let $G$ consist of three triangles $A=\{1,2,3\}$, $B=\{4,5,6\}$, $C=\{7,8,9\}$ and a central node $10$. All triangle edges are deterministic. On the other hand, the three \emph{spokes} (edges $(10,1)$, $(10,4)$, $(10,7)$) are independently present at each step with probabilities $p_A=0.8$, $p_B=0.5$, $p_C=0.2$: that is, spoke $A$ is reliable, and spoke $C$ more often fails. We take $k_{\min}=3$, and let exogenous absorption $\alpha=0.1$ per step, the smallest-index tie-breaker, and the uniform initial distribution $s=\tfrac1{10}\mathbf 1$; see Figure~\ref{fig:crossroads} for a visual representation.

\begin{figure}[htbp]
    \centering
    \begin{tikzpicture}[
        every node/.style={circle, draw, minimum size=6.5mm, inner sep=0pt, font=\small},
        gateway/.style={fill=blue!15},
        interior/.style={fill=gray!10},
        hub/.style={fill=red!20, minimum size=7.5mm},
        spoke/.style={dashed, thick},
        tri/.style={thick}
    ]
        \node[hub] (n10) at (0,0) {10};
        \node[gateway] (n1) at (-2.5, 1.45) {1};
        \node[interior] (n2) at (-4.0, 2.2) {2};
        \node[interior] (n3) at (-4.0, 0.7) {3};
        \draw[tri] (n1) -- (n2); \draw[tri] (n2) -- (n3); \draw[tri] (n1) -- (n3);
        \node[gateway] (n4) at (2.5, 0) {4};
        \node[interior] (n5) at (4.2, 0.8) {5};
        \node[interior] (n6) at (4.2, -0.8) {6};
        \draw[tri] (n4) -- (n5); \draw[tri] (n5) -- (n6); \draw[tri] (n4) -- (n6);
        \node[gateway] (n7) at (-2.5, -1.45) {7};
        \node[interior] (n8) at (-4.0, -0.7) {8};
        \node[interior] (n9) at (-4.0, -2.2) {9};
        \draw[tri] (n7) -- (n8); \draw[tri] (n8) -- (n9); \draw[tri] (n7) -- (n9);
        \draw[spoke] (n10) -- (n1) node[midway, above, draw=none, fill=white, inner sep=1pt, font=\footnotesize] {$p_A=0.8$};
        \draw[spoke] (n10) -- (n4) node[midway, above, draw=none, fill=white, inner sep=1pt, font=\footnotesize] {$p_B=0.5$};
        \draw[spoke] (n10) -- (n7) node[midway, below, draw=none, fill=white, inner sep=1pt, font=\footnotesize] {$p_C=0.2$};
        \node[draw=none, font=\itshape\small] at (-4.0, 2.9) {Triangle $A$};
        \node[draw=none, font=\itshape\small] at (4.2, 1.55) {Triangle $B$};
        \node[draw=none, font=\itshape\small] at (-4.0, -2.9) {Triangle $C$};
    \end{tikzpicture}
    \caption{The asymmetric crossroads toy example. Three deterministic triangles are attached to a central node $10$ via stochastic spokes (dashed edges) with existence probabilities equal to $p_A=0.8$, $p_B=0.5$, $p_C=0.2$.}
    \label{fig:crossroads}
\end{figure}
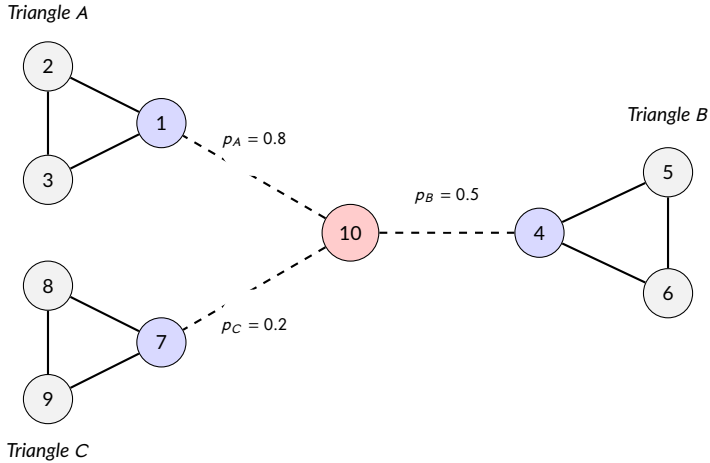

The competing measures behave as follows (all measures normalized to sum to $1$ over $V$ for comparability; exact values are also shown in Figure~\ref{fig:crossroads-comparison}). Classical betwenness on the base graph $G$ cannot distinguish the three gateways at all: the base topology is symmetric in them, so nodes $1,4,7$ tie at $0.203$ each, with node $10$ at $0.391$ and the interior gray nodes at zero. Random-walk betweenness \cite{newman_2005} likewise ties the gateways ($0.167$ each) and, through its degree bias, it also assigns every gray node some positive value ($0.031$). In other words, the structural differentiation between mediating and merely well-connected nodes is not clear anymore. Realization-averaged betweenness, or $\mathbb E_H[\mathrm B_v(H)]$, does separate the gateways, and does so monotonically in spoke reliability: $0.294$, $0.237$, $0.116$ for nodes $1$, $4$, $7$.

AFC, on the other hand, can be computed exactly from the $10$-state chain, and is equal to
$$
b=(0.262,\,0.013,\,0.013,\,0.115,\,0.013,\,0.013,\,0.149,\,0.013,\,0.013,\,0.397)
$$
with $\mathbb E_s[T]=7.78$, giving the ranking $10\succ 1\succ 7\succ 4\succ$ all gray nodes. Three features distinguish it from all references above. 

First, conditional on a valid continuation, the chain remains at node $10$ with probability $0.544$. Node $10$ is the unique betweenness maximizer whenever at least two spokes are present (this happens with probability $0.50$). It is also the unique maximizer with probability $0.42$ when exactly one spoke is present and the chain moves to that gateway. Finally, with probability $0.08$, node $10$ is isolated and the step is invalid. 

Secondly, interior nodes are ``invisible'' to AFC. They are never reported as centers (within an isolated triangle all betweenness values tie at zero and the tie-breaker selects the gateway; in any larger component the gateway or hub strictly dominates), so their AFC mass is exactly the initial-distribution contribution $s_v/\mathbb E_s[T]=0.1/7.78\approx 0.013$. 

Third, we observe a reliability inversion: AFC ranks gateway $7$ ($0.149$) above gateway $4$ ($0.115$) (the opposite of realization-averaged betweenness), while keeping gateway $1$ first. Because spoke $C$ rarely works, anchor $7$ is repeatedly re-reported as the local center of its own isolated triangle, accumulating occupancy that no average of per-realization scores records or sees. Gateway $1$ nevertheless remains first among gateways because the hub transitions occupancy back to it: from anchor $10$, the single-spoke continuations move to nodes $1$, $4$, $7$ with conditional probabilities $0.348$, $0.087$, $0.022$.

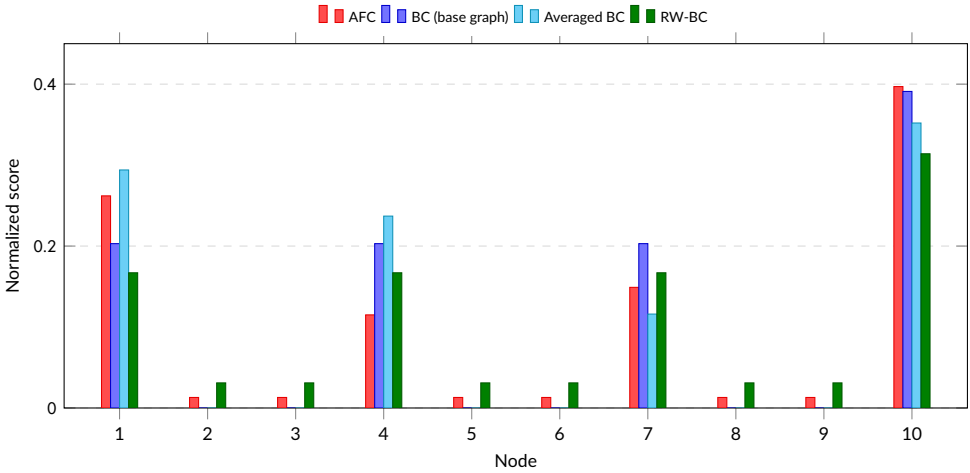
\begin{figure}[htbp]
    \centering
    \begin{tikzpicture}
    \begin{axis}[
        width=0.95\linewidth, height=6.4cm,
        ybar=0pt, bar width=3.4pt,
        enlarge x limits=0.07,
        ymin=0, ymax=0.45,
        ylabel={\footnotesize Normalized score},
        xlabel={\footnotesize Node},
        symbolic x coords={1,2,3,4,5,6,7,8,9,10},
        xtick=data,
        ymajorgrids=true, major grid style={dashed, gray!30},
        legend style={at={(0.5,1.02)}, anchor=south, legend columns=4, font=\scriptsize, draw=none},
        tick label style={font=\footnotesize}
    ]
    \addplot+[fill=red!70, draw=red!90!black] coordinates {
        (1,0.262) (2,0.013) (3,0.013) (4,0.115) (5,0.013) (6,0.013) (7,0.149) (8,0.013) (9,0.013) (10,0.397)};
    \addplot+[fill=blue!55, draw=blue!80!black] coordinates {
        (1,0.203) (2,0.000) (3,0.000) (4,0.203) (5,0.000) (6,0.000) (7,0.203) (8,0.000) (9,0.000) (10,0.391)};
    \addplot+[fill=cyan!50, draw=cyan!70!black] coordinates {
        (1,0.294) (2,0.000) (3,0.000) (4,0.237) (5,0.000) (6,0.000) (7,0.116) (8,0.000) (9,0.000) (10,0.352)};
    \addplot+[fill=green!50!black, draw=green!35!black] coordinates {
        (1,0.167) (2,0.031) (3,0.031) (4,0.167) (5,0.031) (6,0.031) (7,0.167) (8,0.031) (9,0.031) (10,0.314)};
    \legend{AFC, BC (base graph), Averaged BC, RW-BC}
    \end{axis}
    \end{tikzpicture}
    \caption{Exact centrality scores on the asymmetric crossroads of Figure~\ref{fig:crossroads}, normalized to sum to one: (i) AFC with uniform $s$; (ii) classical betweenness on the base graph; (iii) realization-averaged betweenness $\mathbb E_H[\mathrm B_v(H)]$; and (iv) random-walk (current-flow) betweenness \cite{newman_2005} on the base graph.}
    \label{fig:crossroads-comparison}
\end{figure}

Hence, AFC is neither a ``robustified'' variant of betweenness nor an extension or refinement of walk centrality. It measures a different quantity: the long-run frequency with which a node is reported as the local betweenness maximum under the chain's own resampling dynamics. This frequency depends jointly on the probability that a node can be central in its realized component and on the probability that it remains so at the next step. 

\section{Problem Statement}

\subsection{Network states and central node}


Let $G(V,E)$ be a base topology, with $|V|=n$. A network state is any random object $Y$ that produces a working graph $G(Y)$: $Y$ may encode edge availability $A\subseteq E$ and/or positive edge weights (e.g., representing travel times) $\tau\in\mathbb{R}_+^{|E|}$, possibly correlated across edges and over time.

For each realization $Y$, compute weighted shortest-path distances and classical (weighted) betweenness 
$\mathrm{B}_v(Y)$. We also fix a deterministic tie-breaker (e.g., select the smallest-index node among ties) and then we can define the \textit{unique} global central node as in \eqref{eq:central-node}.
\begin{equation}\label{eq:central-node}
c(Y)\ :=\ \arg\max_{v\in V}\ \mathrm{B}_v(Y)\ \in V.
\end{equation}

Let $H(V,E_H,\tau)$ be a realization (working graph) with positive edge weights $\tau_e>0$ for each edge $e\in E$. For $s,t\in V$, let $d_H(s,t)$ denote the weighted shortest-path distance in $H$, with the convention $d_H(s,t)=+\infty$ if $t$ is unreachable from $s$. For $s\neq t$, let $\sigma_{st}(H)$ be the number of weighted shortest paths from $s$ to $t$ in $H$ (so $\sigma_{st}(H)=0$ if $d_H(s,t)=+\infty$), and let $\sigma_{st}(v;H)$ be the number of such shortest paths whose internal vertices include $v$ (excluding endpoints). Define the (non-normalized) weighted betweenness of $v$ on $H$ by
$$
\mathrm{B}_v(H)
:= \sum_{\substack{s,t\in V\\ s\neq t,\ s\neq v,\ t\neq v}}
\mathbf 1\{d_H(s,t)<\infty\}\,\frac{\sigma_{st}(v;H)}{\sigma_{st}(H)}.
$$
When $d_H(s,t)=+\infty$, the pair $(s,t)$ contributes $0$.

\begin{remark}[Normalization]
We use the non-normalized betweenness definition above for ranking, hence normalization constants (e.g., division by $(n-1)(n-2)$ in the connected case, or by the number of reachable pairs) do not affect $c(H)$ as long as they do not depend on $v$.
\end{remark}

\subsection{Component-wise (local) centers for disconnected realizations}
The realization $H$ may be disconnected and hence comprise components. In that case, a single global argmax $c(H)$ is not representative of each component: hence, we can define a component-wise betweenness center anchored at a node.

For a realized working graph $H(V,E_H,\tau)$, let $\mathcal{K}(H)$ denote the collection of components (maximally connected subgraphs) of $H$. For any $i\in V$, let $K(i;H)\in\mathcal{K}(H)$ be the unique component containing $i$ (so $K(i;H)=\{i\}$ if $i$ is isolated). Using the same deterministic tie-breaker as in \eqref{eq:central-node}, define the component-wise center of $K\in\mathcal{K}(H)$ by:
$$
c_{\mathrm{comp}}(K;H)\ :=\ \arg\max_{v\in K}\ \mathrm{B}_v(H)\ \in K.
$$

Given an anchor node $i\in V$ and a realized working graph $H$, define the local center by
\begin{equation}\label{eq:local-center}
c_{\mathrm{loc}}(i;H)\ :=\ c_{\mathrm{comp}}(K(i;H);H)\ \in V.
\end{equation}


We can further fix a threshold $k_{\min}$ and treat realizations with $|K(i;H)|<k_{\min}$ (i.e., small enough) as having no valid continuation from anchor $i$. Motivated by \cite{artime_grassia_de_domenico_gleeson_makse_mangioni_perc_radicchi_2024}, it is possible to set a component size threshold either as an absolute cutoff $k_{\min}$ or as a relative cutoff $k_{\min}=\lceil \theta |V| \rceil$ with $\theta\in \left(0,1\right)$, and treat realizations with $|K(i;H)|<k_{\min}$ as having no valid continuation from anchor $i$; this event is absorbed into $\mathcal A_i$. This avoids tie-break artifacts dominating the dynamics in highly fragmented states.

\subsection{Multiple local centers}

We may opt for a set of nodes with maximum centrality value (top-$k$ nodes), instead of a single center in each component \cite{riondato_kornaropoulos_2015, segarra_ribeiro_2016}.
For a realized working graph $H(V,E_H,\tau)$ and a component $K\in\mathcal{K}(H)$, let $\pi_{K,H}$ denote a deterministic tie-broken ordering of the vertices in $K$ by decreasing $\mathrm{B}_v(H)$ (ties resolved by the same global rule, restricted to $K$). Then, the component-wise top-$k$ set is defined as
\begin{equation}\label{eq:topk-comp}
\mathrm{Top}_k(K;H)\ :=\ \{\pi_{K,H}(1),\dots,\pi_{K,H}(k\wedge |K|)\}\ \subseteq\ K.
\end{equation}

\begin{definition}[Local top-$k$]\label{def:local-topk}
Given an anchor node $i\in V$ and a realized working graph $H$, define the local top-$k$ set by
\begin{equation}\label{eq:local-topk}
\mathrm{Top}^{\mathrm{loc}}_k(i;H)\ :=\ \mathrm{Top}_k\!\big(K(i;H);H\big)\ \subseteq\ V.
\end{equation}
When $|K(i;H)|<k_{\min}$, we treat the step as having no valid continuation from anchor $i$.
\end{definition}

\begin{remark}[Choosing $k$ vs.\ $k_{\min}$]
If one component requires the top-$k$ set to contain exactly $k$ distinct nodes, it suffices to choose $k_{\min}\ge k$. Otherwise \eqref{eq:topk-comp} uses $k\wedge |K|$ and remains well-defined for any $k\ge 1$ whenever $|K|\ge k_{\min}$.
\end{remark}

\section{Absorbing Markov chain on reported local centers}

We summarize the stochastic network evolution by tracking only the reported local center.
The state space is $S=V\cup\{\perp\}$, where $\perp$ is an absorbing state.

If the current state is $i\in V$, we generate one new realized working graph anchored at $i$.
If this realization has no valid continuation, the chain moves to $\perp$.
Otherwise it moves to the tie-broken local center of the realized graph.

Formally, let $(\mathcal H,\mathscr H)$ be the measurable space of realized working graphs.
Fix the following one-step simulator throughout this section:
for each $i\in V$, there exists $(\Omega_i,\mathcal F_i,\mathbb P_i), \Pi_i:\Omega_i\to\mathcal H,$ and $\mathcal A_i\in\mathcal F_i$.

Here $\Pi_i(\omega)$ is the realized working graph produced from anchor $i$, and $\mathcal A_i$ is the event of no valid continuation. In particular, $\{|K(i;\Pi_i(\omega))|<k_{\min}\}\subseteq \mathcal A_i$. On $\Omega_i\setminus\mathcal A_i$, the next reported node is $c_{\mathrm{loc}}(i;\Pi_i(\omega))$.

This gives an absorbing Markov chain (AMC) on $S$ with transition matrix
\begin{equation}\label{eq:amc}
P=\begin{bmatrix}
Q & r \\
0 & 1
\end{bmatrix},
\qquad
Q\in[0,1]^{n\times n},\quad r\in[0,1]^n,\quad Q\mathbf{1}+r=\mathbf{1}.
\end{equation}
$Q$ is the transient block on $V$, and $r_i=P_{i\perp}$ is the absorption probability from $i$. Equivalently, for $i,j\in V$, $P_{ij} = \mathbb P_i\!\big(\omega\notin\mathcal A_i,\ c_{\mathrm{loc}}(i;\Pi_i(\omega))=j\big), P_{i\perp} = \mathbb P_i(\mathcal A_i)$.

The Markov property means that, conditional on the current reported node, the next step does not depend on earlier history. In the present setup, for every $t\ge 0$,
\[
\mathbb P(X_{t+1}=j\mid X_t=i, X_{0:t-1})
=
\mathbb P_i\!\big(\omega\notin\mathcal A_i,\ c_{\mathrm{loc}}(i;\Pi_i(\omega))=j\big),
\]
and $\mathbb P(X_{t+1}=\perp\mid X_t=i, X_{0:t-1}) = \mathbb P_i(\mathcal A_i)$.

Assume $\rho(Q)<1$. Note that this holds automatically whenever every transient state has strictly positive absorption probability, e.g., under an exogenous per-step stopping probability $\alpha>0$, since then $\|Q\|_\infty\le 1-\alpha<1$.

Then, absorption occurs almost surely and the expected absorption time is finite.
Let $X_0\sim s\in\Delta^{n-1}$, and define the absorption time
\begin{equation}\label{eq:absorption-time}
T:=\inf\{t\ge 0:X_t=\perp\}.
\end{equation}
The fundamental matrix is
\begin{equation}\label{eq:fundamental}
N:=(I-Q)^{-1}=\sum_{t\ge 0}Q^t.
\end{equation}
Hence the expected pre-absorption visit counts are
\begin{equation}\label{eq:visits}
\mu(s):=sN,
\end{equation}
and the expected pre-absorption length is
\begin{equation}\label{eq:expected-length}
\mathbb E_s[T]=\mu(s)\mathbf 1.
\end{equation}
We define \textbf{absorbing-frequency centrality} (AFC) as in \eqref{eq:afc-def}.
\begin{equation}\label{eq:afc-def}
b(s):=\frac{\mu(s)}{\mu(s)\mathbf 1}\in\Delta^{n-1}.
\end{equation}
Thus $b_v(s)$ is the fraction of pre-absorption steps (including $t=0$) spent at node $v$.

\begin{remark}[Temporal dependence and state augmentation]
If the process $(X_t)$ is not Markov because the network has additional memory beyond the current reported node, one can enlarge the state space by adding an environment variable $Z_t$. The resulting process is on $(V\times\mathcal Z)\cup\{\perp\}$, and the same occupation and absorption constructions still apply on this enlarged transient space. When $\mathcal Z$ is finite, the matrix formulae above remain unchanged.
\end{remark}

\subsection{Why the AMC state is a single node}
\label{subsec:single-node-uniqueness}

The chain is node-valued, so each valid realization must produce a single next node.
Before tie-breaking, however, a realized graph may have several best candidates.

For $i\in V$ and a valid realized graph $H$ (that is, $|K(i;H)|\ge k_{\min}$), define the unresolved local co-maximizer set $\Gamma^{\max}(i;H) :=
\arg\max_{v\in K(i;H)} \mathrm{B}_v(H)
\subseteq K(i;H)$.

This is the local Top-$1$ output before tie-breaking. More generally, the discussion below applies to any nonempty candidate set $\Gamma(i;H)\subseteq V$, for example $\Gamma^{\max}(i;H)$ or $\mathrm{Top}^{\mathrm{loc}}_k(i;H)$. Fix $i\in V$ and such a candidate correspondence $\Gamma(i;\cdot)$ on valid realizations. Assume $\Gamma(i;\cdot)$ is measurable in the sense that $\{H\in\mathcal H:\ u\in \Gamma(i;H)\}\in\mathscr H$ for every $u\in V$. An admissible selector is a measurable map $g_i:\mathcal H\to V$ such that $g_i(H)\in \Gamma(i;H)$ for every valid $H$. Each admissible selector induces a node-valued row by $P^{g_i}_{ij} := \mathbb P_i\!\big(\omega\notin\mathcal A_i,\ g_i(\Pi_i(\omega))=j\big), j\in V$, and $P^{g_i}_{i\perp}:=1-\sum_{j\in V}P^{g_i}_{ij}$.

\begin{proposition}[Selector-independence and nonuniqueness]
\label{prop:single-node-uniqueness}
Fix $i\in V$ and a measurable candidate correspondence $\Gamma(i;\cdot)$.

\medskip

\noindent (i) If every two admissible selectors agree $\mathbb P_i$-almost surely on $\Omega_i\setminus\mathcal A_i$, then they induce the same node-valued row.

\medskip

\noindent
(ii) If $\mathbb P_i\!\big(|\Gamma(i;\Pi_i(\omega))|\ge 2,\ \omega\notin\mathcal A_i\big)>0$, then there exist two admissible selectors that induce different node-valued rows.
\end{proposition}

\begin{proof}
Part (i) is immediate:
if $g_i$ and $\tilde g_i$ agree $\mathbb P_i$-almost surely on $\Omega_i\setminus\mathcal A_i$, then for every $j\in V$, $\mathbf 1\{\omega\notin\mathcal A_i,\ g_i(\Pi_i(\omega))=j\}
= \mathbf 1\{\omega\notin\mathcal A_i,\ \tilde g_i(\Pi_i(\omega))=j\}$ holds $\mathbb P_i$-almost surely, so $P^{g_i}_{ij}=P^{\tilde g_i}_{ij}$.

For part (ii), let $E_i:=\{\omega\notin\mathcal A_i:\ |\Gamma(i;\Pi_i(\omega))|\ge 2\}$. By assumption, $\mathbb P_i(E_i)>0$. Since $V$ is finite, there exist distinct $u,v\in V$ such that $E_{i,u,v} := \{\omega\in E_i:\ u,v\in\Gamma(i;\Pi_i(\omega))\}$ also has positive $\mathbb P_i$-measure.

Now, fix a strict precedence order $\prec$ on $V$, and define the baseline selector $h_i(H):=\min_{\prec}\Gamma(i;H)$. Because $V$ is finite and the membership events $\{H:u\in\Gamma(i;H)\}$ are measurable, $h_i$ is also measurable. We define two admissible selectors:
$$
g_i^{(u)}(H)
:=
\begin{cases}
u, & u\in\Gamma(i;H),\\
h_i(H), & u\notin\Gamma(i;H),
\end{cases}
\qquad
g_i^{(v)}(H)
:=
\begin{cases}
v, & v\in\Gamma(i;H),\\
h_i(H), & v\notin\Gamma(i;H).
\end{cases}
$$
Both are measurable and satisfy $g_i^{(u)}(H),g_i^{(v)}(H)\in\Gamma(i;H)$ on valid realizations.

Consider the events
\[
F_u^{(u)}
:=
\{\omega\notin\mathcal A_i:\ g_i^{(u)}(\Pi_i(\omega))=u\},
\qquad
F_u^{(v)}
:=
\{\omega\notin\mathcal A_i:\ g_i^{(v)}(\Pi_i(\omega))=u\}.
\]
We claim that $F_u^{(v)}\subseteq F_u^{(u)}$.
Indeed, if $\omega\in F_u^{(v)}$, then $g_i^{(v)}(\Pi_i(\omega))=u$.
This is impossible when $v\in\Gamma(i;\Pi_i(\omega))$, because in that case $g_i^{(v)}$ would output $v$.
Hence $v\notin\Gamma(i;\Pi_i(\omega))$, so $g_i^{(v)}(\Pi_i(\omega))=h_i(\Pi_i(\omega))=u$.
In particular, $u\in\Gamma(i;\Pi_i(\omega))$, and therefore $g_i^{(u)}(\Pi_i(\omega))=u$.
So $F_u^{(v)}\subseteq F_u^{(u)}$.

Moreover, the inclusion is strict on $E_{i,u,v}$: for every $\omega\in E_{i,u,v}$, $g_i^{(u)}(\Pi_i(\omega))=u, g_i^{(v)}(\Pi_i(\omega))=v\neq u$. Thus $E_{i,u,v}\subseteq F_u^{(u)}\setminus F_u^{(v)}$, and since $\mathbb P_i(E_{i,u,v})>0$, $\mathbb P_i(F_u^{(u)})>\mathbb P_i(F_u^{(v)})$. Equivalently, $P^{g_i^{(u)}}_{iu}>P^{g_i^{(v)}}_{iu}$. So the two induced rows are different.
\end{proof}

\begin{corollary}[Deterministic tie-breaking gives a canonical row]
\label{cor:tie-break-unique-row}
Fix a strict total precedence order $\prec$ on $V$ (for example, increasing node index). Then
$$
c_{\mathrm{loc}}(i;H):=\min_{\prec}\Gamma^{\max}(i;H)
$$
is an admissible selector. It therefore defines a canonical node-valued row, and hence a well-defined node-valued AMC on $V\cup\{\perp\}$.
\end{corollary}

\begin{proposition}
\label{prop:set-state-not-unique}
Let $A\subseteq V$ with $|A|\ge 2$.
If there exist $u,v\in A$ such that $P_{u\cdot}\neq P_{v\cdot}$, then $A$ does not determine a unique next-step row on $V\cup\{\perp\}$. For every $\lambda\in[0,1]$, $R^{(\lambda)}_{A\cdot}:=\lambda P_{u\cdot}+(1-\lambda)P_{v\cdot}$ is a valid probability row. Hence propagating the whole candidate set requires extra structure, such as a selector, explicit weights on the elements of $A$, or a larger state space.
\end{proposition}

\begin{proof}
Each $R^{(\lambda)}_{A\cdot}$ is a convex combination of probability rows, so it is again a probability row.
Since $P_{u\cdot}\neq P_{v\cdot}$, at least one coordinate differs, and therefore $R^{(\lambda)}_{A\cdot}$ depends nontrivially on $\lambda$.
\end{proof}

Proposition~\ref{prop:set-state-not-unique} shows that unresolved candidate sets cannot in general be propagated within the same node-valued AMC. Unless all members of the set have identical outgoing rows, the set does not determine a unique transition law. Thus, some fixed selection rule is required: an unresolved candidate set does not by itself determine a transition law. Deterministic tie-breaking is a sufficient and easily implementable choice, but, as the next remark records, it is not the only one.

\begin{remark}[Randomized tie-breaking and label-invariance]\label{rem:randomized-tiebreak}
Proposition~\ref{prop:set-state-not-unique} shows that propagating an unresolved set is ambiguous; it does not show that the selector must be deterministic. A randomized selector that draws uniformly from the candidate set also induces a unique node-valued row,
\[
P_{ij}\;=\;\mathbb E_i\!\left[\mathbf 1\{\omega\notin\mathcal A_i\}\,
\frac{\mathbf 1\{j\in\Gamma(i;\Pi_i(\omega))\}}{|\Gamma(i;\Pi_i(\omega))|}\right],
\qquad j\in V,
\]
and hence a well-defined AMC, since the auxiliary randomization can be absorbed into the one-step probability space $(\Omega_i,\mathcal F_i,\mathbb P_i)$. The two conventions differ in one important respect: the $\min_\prec$ rule makes AFC depend on the node labeling, because in realizations with symmetric co-maximizers (e.g., an isolated clique, where all betweenness values tie) the rule systematically routes occupancy to low-index nodes. The uniform rule is equivariant under relabeling and removes this artifact, at the cost of a small amount of additional simulator randomness. 
\end{remark}

Proposition~\ref{prop:single-node-uniqueness} does not state that only the highest-ranked node is meaningful. It only states that an unresolved candidate set does not by itself specify a unique node-valued AMC. For any fixed rank $r$, one may instead define $c_{\mathrm{loc}}^{(r)}(i;H):=\pi_{K(i;H),H}(r)$ whenever $r\le |K(i;H)|$, and study the corresponding node-valued process. Alternatively, one may keep the full Top-$k$ output and work with a set-valued or tuple-valued state space.

\subsection{Well-posedness: order/branch invariance and row normalization}
\label{subsec:wellposed}

We may reach the same realized graph through a different reveal order. 
The compressed transition row should therefore depend only on the law of the realized graph and on the absorption rule, not on how that graph came to be generated.

\begin{definition}[Pushforward of the one-step simulator]
\label{def:pushforward}
First, fix $i\in V$. Then, the pushforward distribution of the realized graph is $\nu_i:=\mathbb P_i\circ\Pi_i^{-1}$.
\end{definition}

\begin{figure}[t]
    \centering
    \includegraphics[width=0.8\linewidth]{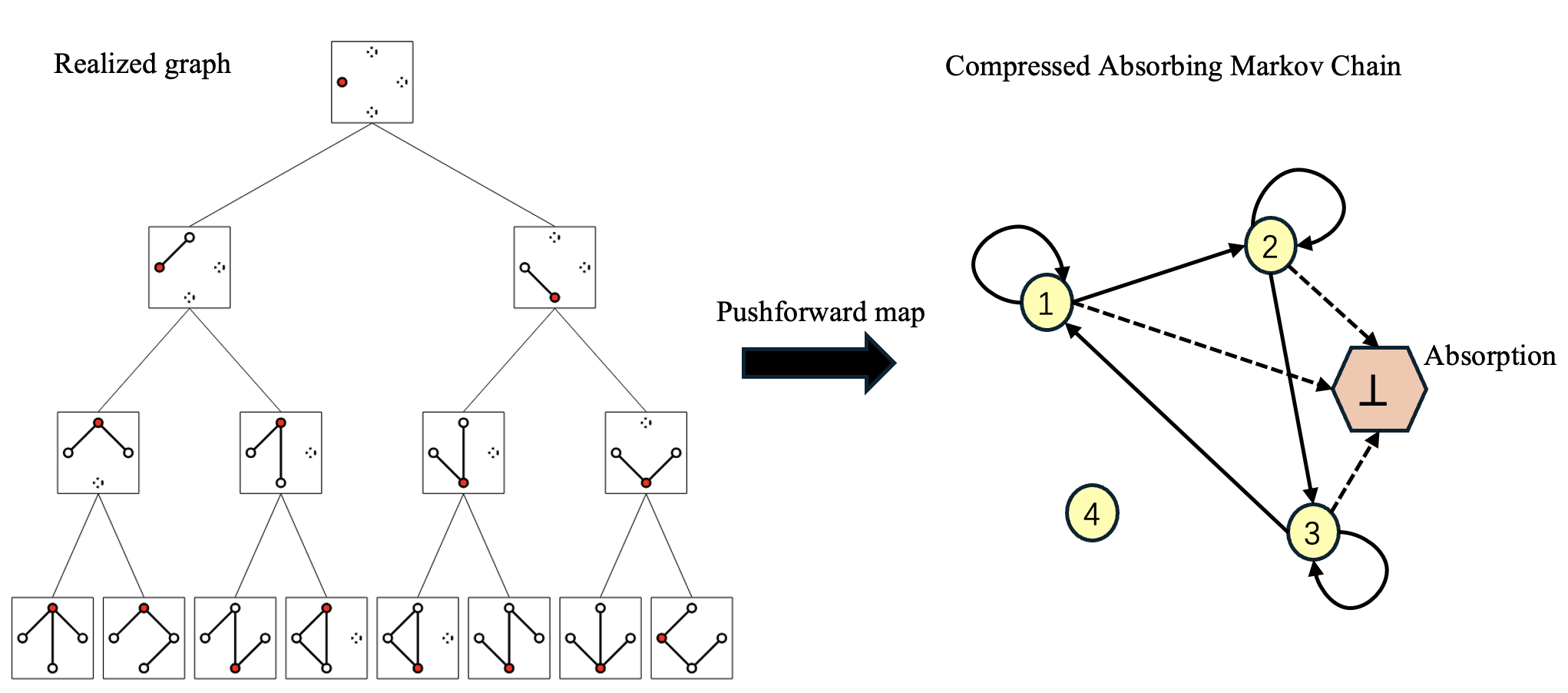}
    \caption{Illustration of the pushforward compression from the branching one-step simulator to the node-valued absorbing Markov chain. For a fixed anchor $i$, each branch yields a realized working graph $H=\Pi_i(\omega)$. Valid realizations are mapped to the tie-broken local center $c_{\mathrm{loc}}(i;H)\in V$, while realizations in $\mathcal A_i$ are mapped to the absorbing state $\perp$. Carrying out this compression for each anchor yields the compressed AMC on $V\cup\{\perp\}$. Solid arrows denote transient transitions and dashed arrows denote absorption; only selected positive-probability transitions are shown.}
    \label{fig:pushforward-amc}
\end{figure}

Figure~\ref{fig:pushforward-amc} gives a schematic view of the compression underlying Definition~\ref{def:pushforward}. For a fixed current anchor $i$, the internal randomness of the one-step simulator can be represented by a branching realization tree, whose terminal branches produce realized working graphs $H=\Pi_i(\omega)$. The detailed branch structure is then collapsed by sending each valid realization to its reported local center $c_{\mathrm{loc}}(i;H)$, and each invalid realization in $\mathcal A_i$ to the absorbing state $\perp$. The resulting description is node-valued and forgets the reveal order itself. The next proposition makes this precise by showing that the compressed transition law is determined by the pushforward law $\nu_i$ of realized graphs together with the conditional absorption function $a_i(\cdot)$, rather than by the particular branching representation used to generate the same graph law.

\begin{proposition}[Order/branch invariance of the compressed step]
\label{prop:order-invariance}
For $i,j\in V$, define
\begin{equation}\label{eq:def-Pij}
P_{ij}
:=
\mathbb P_i\!\big(\omega\notin\mathcal A_i,\ c_{\mathrm{loc}}(i;\Pi_i(\omega))=j\big),
\qquad
P_{i\perp}
:=
\mathbb P_i(\mathcal A_i).
\end{equation}
Let $a_i(H):=\mathbb E_i[\mathbf 1_{\mathcal A_i}\mid \Pi_i=H]$, that is, a measurable version of the conditional absorption probability given the realized graph. Then the row $(P_{ij})_{j\in V\cup\{\perp\}}$ is determined by $\nu_i$ together with $a_i(\cdot)$. Hence any two generation procedures that induce the same $\nu_i$ and the same conditional absorption function $a_i(\cdot)$ give the same AMC row for $i$. If $\mathcal A_i\in\sigma(\Pi_i)$, then $a_i(H)\in\{0,1\}$ $\nu_i$-almost surely, and the row depends only on $\nu_i$ and the graph-based absorption rule.
\end{proposition}

\begin{proof}
Since $c_{\mathrm{loc}}(i;\Pi_i(\omega))$ depends only on the realized graph $H=\Pi_i(\omega)$, the tower property gives
\begin{align*}
P_{ij} =
\mathbb E_i\!\Big[
\mathbf 1\{c_{\mathrm{loc}}(i;\Pi_i)=j\}\,
\mathbf 1\{\omega\notin\mathcal A_i\}
\Big] 
&=
\mathbb E_i\!\Big[
\mathbf 1\{c_{\mathrm{loc}}(i;\Pi_i)=j\}\,
\mathbb E_i[\mathbf 1_{\Omega_i\setminus\mathcal A_i}\mid \Pi_i]
\Big]\\ 
&=
\int_{\mathcal H}
\mathbf 1\{c_{\mathrm{loc}}(i;H)=j\}\,
\bigl(1-a_i(H)\bigr)\,d\nu_i(H),
\end{align*}
where $a_i(H)=\mathbb P_i(\mathcal A_i\mid \Pi_i=H)\in[0,1]$. Thus $P_{ij}$ depends only on $\nu_i$ and $a_i(\cdot)$. Similarly, $P_{i\perp}=\int_{\mathcal H} a_i(H)\,d\nu_i(H)$.
\end{proof}

\begin{lemma}[Row normalization]
\label{lem:row-normalization}
For each $i\in V$, $P_{ii}+\sum_{j\in V,\ j\neq i}P_{ij}+P_{i\perp}=1$.
\end{lemma}

\begin{proof}
The events $E_{ij}:=\{\omega\notin\mathcal A_i,\ c_{\mathrm{loc}}(i;\Pi_i(\omega))=j\}, j\in V$, together with $E_{i\perp}:=\mathcal A_i$, are pairwise disjoint and their union is $\Omega_i$. Therefore,
$$
1=\mathbb P_i(\Omega_i)
=\sum_{j\in V}\mathbb P_i(E_{ij})+\mathbb P_i(E_{i\perp})
=\sum_{j\in V}P_{ij}+P_{i\perp}.
$$
\end{proof}
\subsection{General equivalence via uniform pre-absorption sampling}
\label{subsec:uniform-step}
We now give a trajectory interpretation of AFC.
The idea is simple: AFC is the distribution of a uniformly chosen pre-absorption step, but under a length-biased law on sample paths.

\begin{definition}[Length-biased uniform pre-absorption step]
\label{def:lengthbiased}
Let $(X_t)_{t\ge 0}$ be any process on $V\cup\{\perp\}$ with absorption time $T$ as in \eqref{eq:absorption-time}. Assume that $T<\infty$ almost surely and $\mathbb E_s[T]<\infty$ under $X_0\sim s$.
Define $\widehat{\mathbb P}_s$ on $(\Omega\times\mathbb Z_{\ge 0},\,\mathcal F\otimes 2^{\mathbb Z_{\ge 0}})$ \cite{gill_vardi_wellner_1988}:
\begin{equation}\label{eq:lengthbiased-law}
\widehat{\mathbb P}_s(B)
:=
\frac{1}{\mathbb E_s[T]}
\,
\mathbb E_s\!\left[
\sum_{t=0}^{T-1}\mathbf 1\{(\omega,t)\in B\}
\right],
\qquad
B\in\mathcal F\otimes 2^{\mathbb Z_{\ge 0}}.
\end{equation}
Equivalently, first sample a path with probability proportional to its pre-absorption length $T(\omega)$, and then choose $U\sim\mathrm{Unif}\{0,\dots,T(\omega)-1\}$. Under $\widehat{\mathbb P}_s$, the random node $X_U$ is a uniformly chosen pre-absorption step.
Moreover, $\widehat{\mathbb P}_s(U=t)=\frac{\mathbb P_s(T>t)}{\mathbb E_s[T]}$.
\end{definition}

\begin{proposition}[Uniform-step equivalence]
\label{prop:uniform-step}
Let $\mu_v(s):=\mathbb E_s\!\Big[\sum_{t=0}^{T-1}\mathbf 1\{X_t=v\}\Big], b_v(s):=\frac{\mu_v(s)}{\mathbb E_s[T]}$. Then, under $\widehat{\mathbb P}_s$ from Definition~\ref{def:lengthbiased}, for every $v\in V$,
\begin{equation}\label{eq:uniform-step}
b_v(s) = \widehat{\mathbb P}_s(X_U=v) =
\frac{\sum_{t\ge 0}\mathbb P_s(X_t=v,\ T>t)}
{\sum_{t\ge 0}\mathbb P_s(T>t)}.
\end{equation}
\end{proposition}

\begin{proof}
Apply Definition~\ref{def:lengthbiased} to the event $\{(\omega,t): X_t=v\}$. This gives
$$
\widehat{\mathbb P}_s(X_U=v) = \frac{1}{\mathbb E_s[T]} \, \mathbb E_s\!\Big[\sum_{t=0}^{T - 1}\mathbf 1\{X_t=v\}\Big] = \frac{\mu_v(s)}{\mathbb E_s[T]} = b_v(s).
$$
Also, $\mu_v(s)=\sum_{t\ge 0}\mathbb P_s(X_t=v,\ T>t), \mathbb E_s[T]=\sum_{t\ge 0}\mathbb P_s(T>t)$, which yields \eqref{eq:uniform-step}.
\end{proof}

\begin{corollary}[Survival-weighted decomposition]
\label{cor:survival-decomp}
Let $w_t:=\frac{\mathbb P_s(T>t)}{\mathbb E_s[T]}$. Whenever $\mathbb P_s(T>t)>0$, define $\pi_t(v):=\mathbb P_s(X_t=v\mid T>t)$. When $\mathbb P_s(T>t)=0$, define $\pi_t$ arbitrarily.
Then $w_t\ge 0$, $\sum_{t\ge 0}w_t=1$, and
\begin{equation}\label{eq:b-survival-average}
b_v(s)=\sum_{t\ge 0} w_t\,\pi_t(v).
\end{equation}
\end{corollary}

\begin{proof}
From \eqref{eq:uniform-step}, $b_v(s)
=
\frac{\sum_{t\ge 0}\mathbb P_s(T>t)\,\mathbb P_s(X_t=v\mid T>t)}
{\mathbb E_s[T]}
=
\sum_{t\ge 0} w_t\,\pi_t(v)$, where terms with $\mathbb P_s(T>t)=0$ have weight $w_t=0$ and therefore do not matter.
\end{proof}

\begin{remark}[Within-path uniform sampling vs.\ length bias]
If one samples a path under the original law $\mathbb P_s$ and only then chooses $U\sim\mathrm{Unif}\{0,\dots,T-1\}$, the resulting distribution is $\mathbb E_s\!\Big[\frac{1}{T}\sum_{t=0}^{T-1}\mathbf 1\{X_t=v\}\Big]$, which need not equal $b_v(s)$. The length bias in \eqref{eq:lengthbiased-law} is therefore essential for \eqref{eq:uniform-step}.
\end{remark}

\section{Row-wise perturbations of the AMC kernel: perturbed AFC and comparison to the mixture formula}
\label{subsec:robust-perturb-b}

Subsection~\ref{subsec:stochastic} gives us a useful benchmark: when all post-initial survival-conditional laws are equal to a common distribution $p$, AFC collapses to the simple mixture formula \eqref{eq:mixture-general}. That benchmark is informative, but it relies on a strong assumption. In many settings, the induced AMC kernel is only approximate: it may come from Monte Carlo construction, empirical fitting, or anchor-dependent local-center dynamics. This makes it necessary to study how AFC behaves when the transient kernel is perturbed.

The purpose of this section is therefore twofold. First, we define AFC for every admissible row-wise perturbation of the AMC kernel and show that the absorbing construction remains well-posed. Second, we explain why such perturbations typically destroy the structure behind Theorem~\ref{thm:mixture}, so that the full survival-weighted representation must be used in place of the simple mixture form.

\subsection{Row-wise additive uncertainty set}

Let $P^0$ be a nominal AMC kernel as in \eqref{eq:amc}, with transient block $Q^0$. Fix row-wise perturbation radii $\varepsilon_{ij}\ge 0$ for $i,j\in V$, together with leak lower bounds $\underline r_i\in(0,1]$. For each row $i$, define
{\small\[
\mathcal U_i^{\pm}
:=
\left\{
q_i\in\mathbb R_+^{V}:\ 
q_{ij}\in\bigl[(Q^0_{ij}-\varepsilon_{ij})_+,\ \min\{Q^0_{ij}+\varepsilon_{ij},1\}\bigr]\ \ \forall j\in V,\ 
\sum_{j\in V}q_{ij}\le 1-\underline r_i
\right\},
\]}
where $(x)_+:=\max\{x,0\}$.
Assume $\mathcal U_i^{\pm}\neq\varnothing$ for every $i$; for example, it is enough that $\sum_{j\in V}(Q^0_{ij}-\varepsilon_{ij})_+\le 1-\underline r_i$.

Each $q_i\in\mathcal U_i^{\pm}$ is lifted to a probability row on $S=V\cup\{\perp\}$ by setting $P_{ij}:=q_{ij}\quad (j\in V)$, and $P_{i\perp}:=1-\sum_{j\in V}q_{ij}$. Thus $P_{i\perp}\ge \underline r_i$, so every admissible perturbation retains a strictly positive absorption probability. If $r_i^0\ge \underline r_i$ for all $i$, then the nominal kernel $P^0$ itself belongs to the uncertainty set.

The global uncertainty set is
$$
\mathcal U^{\pm}
:=
\left\{
P=
\begin{bmatrix}Q&r\\0&1\end{bmatrix}:\ 
\forall i\in V,\ (Q_{ij})_{j\in V}\in\mathcal U_i^{\pm},\ \ r_i=1-\sum_{j\in V}Q_{ij}
\right\}.
$$

\begin{remark}[Uniform transience under leak lower bounds]
\label{rem:pm-transience}
Let $\underline r_{\min}:=\min_{i\in V}\underline r_i$.
For any $P\in\mathcal U^{\pm}$, $\|Q\|_\infty\le 1-\underline r_{\min}<1$, hence $\rho(Q)<1$.
Therefore the fundamental matrix $N(P)$ exists for every admissible kernel, and
\begin{equation}\label{eq:pm-N-bound}
\|N(P)\|_\infty
\le
\sum_{t\ge 0}\|Q\|_\infty^t
\le
\frac{1}{\underline r_{\min}}.
\end{equation}
So the absorbing construction is uniformly well posed over the whole uncertainty set.
\end{remark}

\subsection{AFC under a fixed perturbed kernel}

Fix $P\in\mathcal U^{\pm}$ and write $Q=Q(P)$ for its transient block.
Let $\mathbb P_s^P$ and $\mathbb E_s^P$ denote probabilities and expectations for the AMC with kernel $P$ started from $X_0\sim s$. As before, $N(P):=(I-Q)^{-1}, \mu(s;P):=sN(P), b(s;P):=\frac{\mu(s;P)}{\mu(s;P)\mathbf 1}$.

For this fixed kernel, $\mathbb P_s^P(X_t=\cdot,\ T>t)=sQ^t, \mathbb P_s^P(T>t)=sQ^t\mathbf 1$. Hence the survival-conditional law at time $t$ is $\pi_t(\cdot;P) := \mathbb P_s^P(X_t=\cdot\mid T>t)
=
\frac{sQ^t}{sQ^t\mathbf 1},
\quad t\ge 0$, with arbitrary $\pi_t(\cdot;P)$ when $sQ^t\mathbf 1=0$. Accordingly,
\begin{equation}\label{eq:pm-survival-average}
b(s;P)
=
\sum_{t\ge 0} w_t(P)\,\pi_t(\cdot;P),
\qquad
w_t(P)
:=
\frac{sQ^t\mathbf 1}{\sum_{k\ge 0}sQ^k\mathbf 1}
=
\frac{\mathbb P_s^P(T>t)}{\mathbb E_s^P[T]}.
\end{equation}
Thus, even after perturbation, AFC remains the survival-weighted average of the conditional laws
$\pi_t(\cdot;P)$.

Why is it that $b(s;P)$ generally does not reduce to the earlier derived mixture formula? Theorem~\ref{thm:mixture} requires a single post-initial law $p$ such that $\pi_t=p,~\text{for all }t\ge 1$. For a fixed AMC kernel $P$, this becomes:
\begin{equation}\label{eq:pm-S1-amc}
\frac{sQ^t}{sQ^t\mathbf 1}=p
,~\text{for all }t\ge 1.
\end{equation}
This is a strong structural constraint: every survival-conditional law after time $0$ must coincide with the same distribution $p$.

Row-wise $\pm$ perturbations typically destroy that structure by separating the rows of $Q$.
Then the family $\{\pi_t(\cdot;P)\}_{t\ge 1}$ generally varies with $t$ and depends on the initial distribution $s$. In this regime AFC no longer collapses to the two-point mixture \eqref{eq:mixture-general}; one must instead use the full survival-weighted representation \eqref{eq:pm-survival-average}, or equivalently the fundamental-matrix formula.

\begin{remark}[Quantifying deviation from post-initial stationarity]
\label{rem:pm-beyond-KL}
Fix any reference distribution $\bar p\in\Delta^{|V|-1}$ and define the mixture proxy
$
b_{\mathrm{mix}}(s;P,\bar p)
:=
w_0(P)\,s+(1-w_0(P))\,\bar p,
\quad
w_0(P)=\frac{1}{\mathbb E_s^P[T]}.
$
Using \eqref{eq:pm-survival-average} and $\pi_0(\cdot;P)=s$, we obtain
$
b(s;P)-b_{\mathrm{mix}}(s;P,\bar p)
=
\sum_{t\ge 1} w_t(P)\bigl(\pi_t(\cdot;P)-\bar p\bigr).
$
Therefore, for any norm $\|\cdot\|$ on $\mathbb R^V$,
$
\|b(s;P)-b_{\mathrm{mix}}(s;P,\bar p)\|
\le
\sum_{t\ge 1} w_t(P)\,\|\pi_t(\cdot;P)-\bar p\|.
$

More generally, let $D_{\mathcal F}(\nu,\eta):=\sup_{f\in\mathcal F}|\nu f-\eta f|$ be an integral probability metric (IPM) \cite{muller1997integral}. Then $D_{\mathcal F}\!\bigl(b(s;P),b_{\mathrm{mix}}(s;P,\bar p)\bigr)
\le
\sum_{t\ge 1} w_t(P)\, D_{\mathcal F}\!\bigl(\pi_t(\cdot;P),\bar p\bigr)$.

In particular, if $V$ is endowed with a ground metric $d$, then the $1$-Wasserstein distance $W_1^d$ is an IPM with $\mathcal F$ the class of $1$-Lipschitz functions, so $$W_1^d\!\bigl(b(s;P),b_{\mathrm{mix}}(s;P,\bar p)\bigr)
\le
\sum_{t\ge 1} w_t(P)\, W_1^d\!\bigl(\pi_t(\cdot;P),\bar p\bigr).
$$

These bounds show clearly that the error from replacing the whole time-varying family
$\{\pi_t(\cdot;P)\}_{t\ge 1}$ by a single law $\bar p$ is controlled by the survival-weighted average discrepancy
between them. The same viewpoint applies to total variation or Wasserstein discrepancies
\cite{rahimian_2018,mohajerin_esfahani_kuhn_2017}.
Information-theoretic quantities such as Kullback--Leibler divergence can also be useful as diagnostics
\cite{ben-tal_den_hertog_de_waegenaere_melenberg_rennen_2013}, but they are not IPMs and therefore do not enter the
preceding bound in the same way.
\end{remark}

\subsection{State-independent absorption}

Let us now study a useful special case. Specifically, assume $P_{i\perp}\equiv \alpha\in(0,1]$ for all $i \in V$. Then $Q=(1-\alpha)\,M$, where $M$ is row-stochastic on $V$ \cite{gleich_2015}. In this case,
\begin{equation}\label{eq:pm-discounted-occupancy}
b(s;P)
=
\frac{s\sum_{t\ge 0}Q^t}{s\sum_{t\ge 0}Q^t\mathbf 1}
=
\alpha\sum_{t\ge 0}(1-\alpha)^t\,sM^t,
\end{equation}
because $M^t\mathbf 1=\mathbf 1$ implies $\sum_{t\ge 0}sQ^t\mathbf 1 = \sum_{t\ge 0}(1-\alpha)^t = \frac{1}{\alpha}$. So in the constant-hazard case, AFC is a geometrically discounted average of the nonabsorbing evolution under $M$.

The mixture benchmark \eqref{eq:mixture-general} is recovered whenever the post-initial laws are constant, namely when $sM^t=p$ for all $t\ge 1$. Sufficient conditions include $M=\mathbf 1 p$, or more generally $sM=p$ and $pM=p$. Without such post-initial stationarity, \eqref{eq:pm-discounted-occupancy} remains the correct representation.

\subsection{First-order sensitivity of AFC}

We are also interested in quantifying AFC's sensitivity to a small perturbation of $Q$. Let $Q=Q^0+E, \rho(Q)<1$, and let $N^0,\mu^0,b^0$ be the quantities induced by $Q^0$ through \eqref{eq:fundamental}, \eqref{eq:visits}, and \eqref{eq:afc-def} \cite{ipsen_meyer_1994}.
The resolvent identity gives
$
N-N^0
=
NEN^0
=
N^0EN.
$
If $E$ is small, for example if $\|N^0E\|<1$ in a submultiplicative norm, then
$
N
=
N^0+N^0EN^0+O(\|E\|^2).
$
Hence
$
\mu-\mu^0
\approx
sN^0EN^0.
$
Since
$
b=\frac{\mu}{\mu\mathbf 1},
$
the corresponding first-order change in AFC is
$$
b-b^0
\approx
\frac{sN^0EN^0}{\mu^0\mathbf 1}
-
b^0\,\frac{sN^0EN^0\mathbf 1}{\mu^0\mathbf 1}.
$$
This gives a local sensitivity approximation for the effect of a small perturbation of the transient kernel. The expansion can be made rigorous with an explicit remainder: from $I-Q=(I-Q^0)(I-N^0E)$ we obtain $N=(I-N^0E)^{-1}N^0$, hence, whenever $\|N^0E\|<1$ in a submultiplicative norm, $N-N^0-N^0EN^0=(N^0E)^2N$ and therefore
\[
\bigl\|N-N^0-N^0EN^0\bigr\|\;\le\;\frac{\|N^0E\|^2\,\|N^0\|}{1-\|N^0E\|},
\]
which bounds the $O(\|E\|^2)$ term explicitly.

\begin{remark}[Robust optimization viewpoint]
\label{rem:pm-robust-view}
Given $\mathcal U^{\pm}$, one may report envelopes $\inf_{P\in\mathcal U^{\pm}} b_v(s;P), \sup_{P\in\mathcal U^{\pm}} b_v(s;P)$,
or optimize a scalar objective over $\mathcal U^{\pm}$, obtain an adversarial kernel $P^\star$, and then report $b(s;P^\star)$. The point is that departures from Theorem~\ref{thm:mixture} are driven by the generic failure of \eqref{eq:pm-S1-amc} under row-wise perturbations.
\end{remark}

\subsection{Ranking reversals and Top-\texorpdfstring{$\boldsymbol{k}$}{k} instability}
\label{sec:rankreversals}

Under row-wise $\pm$ perturbations of $Q$ the AFC ordering may change, including Top-$k$ sets, and may produce ranking reversals \cite{wills_ipsen_2009, bressan_peserico_2010, chartier_kreutzer_langville_pedings_2011, segarra_ribeiro_2016}.
For $u,v\in V$, define the visit-gap
\begin{equation}\label{eq:pm-gap}
G_{uv}(P)
:=
\mu_u(s;P)-\mu_v(s;P)
=
\bigl(sN(P)\bigr)_u-\bigl(sN(P)\bigr)_v,
\end{equation}
where $N(P)$ is the fundamental matrix associated with $Q(P)$.
Since $b(s;P)$ is a positive normalization of $\mu(s;P)$,
$
\operatorname{sign}\bigl(b_u(s;P)-b_v(s;P)\bigr)
=
\operatorname{sign}\bigl(G_{uv}(P)\bigr).
$

A convenient robust diagnostic is 
$$
\underline G_{uv}
:=
\inf_{P\in\mathcal U^{\pm}}G_{uv}(P), \qquad \overline G_{uv}
:=
\sup_{P\in\mathcal U^{\pm}}G_{uv}(P).
$$ 
If $\underline G_{uv}>0$, then $u$ robustly outranks $v$ over all admissible perturbations. If $\overline G_{uv}<0$, then $v$ robustly outranks $u$. If $0\in[\underline G_{uv},\overline G_{uv}]$, then the interval diagnostic does not certify a strict order and indicates possible ranking instability. The same logic applies to Top-$k$ stability by checking boundary pairs between the nominal Top-$k$ set and its complement.

Now assume $P^0\in\mathcal U^{\pm}$ and set $\underline r_{\min}:=\min_i\underline r_i$, and
$\bar\varepsilon:=\max_{i\in V}\sum_{j\in V}\varepsilon_{ij}$. Then $\|Q(P)-Q^0\|_\infty\le \bar\varepsilon$ for all $P\in\mathcal U^{\pm}$. Using $N(P)-N^0=N(P)\,(Q(P)-Q^0)\,N^0, $ together with \eqref{eq:pm-N-bound}, we obtain $\|\mu(s;P)-\mu(s;P^0)\|_\infty
=
\|s(N(P)-N^0)\|_\infty
\le
\frac{\bar\varepsilon}{\underline r_{\min}^2}.
$
Therefore, for all $u,v\in V$,
$
|G_{uv}(P)-G_{uv}(P^0)|
\le
2\,\|\mu(s;P)-\mu(s;P^0)\|_\infty
\le
\frac{2\bar\varepsilon}{\underline r_{\min}^2}.
$

Consequently, a sufficient certificate for preserving the nominal order $u\succ v$ under all admissible perturbations is $G_{uv}(P^0)>\frac{2\bar\varepsilon}{\underline r_{\min}^2}$.

If this inequality fails, the certificate is inconclusive: the nominal order may still be robust, but one can no longer guarantee it from this bound alone, and a sharper interval computation or direct optimization over $\mathcal U^{\pm}$ is then needed.

The certificate above relies only on the worst-case bound \eqref{eq:pm-N-bound} and is therefore loose. The same resolvent identity yields a pair-specific certificate that replaces the crude factor $2/\underline r_{\min}^{\,2}$ with a column difference of the nominal fundamental matrix, available at no extra cost.

\begin{proposition}[Sharpened pair-specific certificate]\label{prop:sharp-certificate}
Let $P^0\in\mathcal U^{\pm}$, $\underline r_{\min}:=\min_i\underline r_i$, and $\bar\varepsilon:=\max_{i\in V}\sum_{j\in V}\varepsilon_{ij}$. Then for every $P\in\mathcal U^{\pm}$ and every $u,v\in V$,
\begin{equation}\label{eq:sharp-gap}
\bigl|G_{uv}(P)-G_{uv}(P^0)\bigr|
\;\le\;
\frac{\bar\varepsilon}{\underline r_{\min}}\,
\bigl\|N^0(e_u-e_v)\bigr\|_\infty
\;=\;
\frac{\bar\varepsilon}{\underline r_{\min}}\,
\max_{i\in V}\bigl|N^0_{iu}-N^0_{iv}\bigr|.
\end{equation}
Consequently, $G_{uv}(P^0)>\frac{\bar\varepsilon}{\underline r_{\min}}\|N^0(e_u-e_v)\|_\infty$ certifies $u\succ v$ over all of $\mathcal U^{\pm}$. Since $0\le N^0_{iw}\le 1/\underline r_{\min}$ entrywise, \eqref{eq:sharp-gap} improves on the earlier threshold $2\bar\varepsilon/\underline r_{\min}^{\,2}$ by at least a factor of two, and by much more whenever the $u$- and $v$-columns of $N^0$ are close.
\end{proposition}

\begin{proof}
By the resolvent identity, $G_{uv}(P)-G_{uv}(P^0)=s\bigl(N(P)-N^0\bigr)(e_u-e_v)=\bigl(sN(P)\,\Delta Q\bigr)\,N^0(e_u-e_v)$ with $\Delta Q:=Q(P)-Q^0$. For a row vector $x$ and a matrix $A$, $\|xA\|_1\le\|x\|_1\|A\|_\infty$; here $\|sN(P)\|_1=sN(P)\mathbf 1=\mathbb E^P_s[T]\le 1/\underline r_{\min}$ and $\|\Delta Q\|_\infty\le\bar\varepsilon$. The claim follows from $|yz|\le\|y\|_1\|z\|_\infty$ applied to $y=sN(P)\Delta Q$ and $z=N^0(e_u-e_v)$.
\end{proof}

In Appendix \ref{appendixB}, we show a worked illustration for the interested reader.




\section{Multi-reward absorbing--frequency centrality for valued local Top-\texorpdfstring{$\boldsymbol{k}$}{k} betweenness candidates}
\label{subsec:multi-node-reward}

The previous section studied how AFC changes when the AMC kernel is perturbed. Here we keep the same AMC and instead change what is rewarded. This distinction is important in applications: uncertainty may enter through the kernel, but the decision objective may also depend on what value is attached to each realized step of the process.

The basic AFC vector $b(s)$ in \eqref{eq:afc-def} assigns unit reward to each pre-absorption visit of the reported node. That is appropriate when one wants to summarize where the node-valued AMC spends its time. In many applications, however, the object of interest is richer than the single reported node: what matters is the total value carried by the entire local Top-$k$ candidate set produced by the realized graph at each step. Examples include total capacity, cumulative risk, aggregate demand, or combined criticality of the current local betweenness candidates.

This section is needed precisely to capture such objectives without enlarging the AMC state space. The key idea is to attach rewards to the same one-step simulator used to construct the AMC in Subsection~\ref{subsec:wellposed}. This lets us score the whole realized local Top-$k$ set at each step while preserving both the fundamental-matrix formula and the length-biased uniform pre-absorption interpretation from Proposition~\ref{prop:uniform-step} and Corollary~\ref{cor:survival-decomp}.

\subsection{Step rewards on the one-step simulator}

For each transient state $i\in V$, let $\omega_t$ denote the simulator draw used at time $t$ when $X_t=i$. Thus, conditional on $X_t=i$, we have $\omega_t\sim\mathbb P_i$, and the next state is generated by the same one-step simulator as in Subsection~\ref{subsec:wellposed}. Define
$$
\mathrm{Next}_i(\omega):=
\begin{cases}
c_{\mathrm{loc}}(i;\Pi_i(\omega)), & \omega\notin\mathcal A_i,\\
\perp, & \omega\in\mathcal A_i.
\end{cases}
$$
So $\mathrm{Next}_i(\omega)$ is the next reported state generated from anchor $i$.

A step reward is a family $\ell=(\ell_i)_{i\in V}$ of measurable maps $\ell_i:\Omega_i\to\mathbb R_+, i\in V$, with finite expectations $\psi_i:=\mathbb E_i[\ell_i(\omega)]<\infty, i\in V$. Write $\psi=(\psi_i)_{i\in V}\in\mathbb R_+^V$.
Along an absorbed trajectory, define the accumulated pre-absorption reward by
$
R^{(\ell)}:=\sum_{t=0}^{T-1}\ell_{X_t}(\omega_t),
$
where $T$ is the absorption time from \eqref{eq:absorption-time}.

\begin{proposition}[Reward-AFC on the same AMC]
\label{prop:reward-afc-general}
For any nonnegative integrable step reward $\ell$,
\begin{equation}\label{eq:mr-afc-def}
b_\ell(s)
:=
\frac{\mathbb E_s[R^{(\ell)}]}{\mathbb E_s[T]}
=
\frac{sN\psi}{sN\mathbf 1}.
\end{equation}
Moreover, if $U$ is the length-biased uniform pre-absorption step from~\ref{def:lengthbiased}, then
\begin{equation}\label{eq:mr-afc-uniform}
b_\ell(s)=\widehat{\mathbb E}_s\!\big[\ell_{X_U}(\omega_U)\big].
\end{equation}
\end{proposition}

\begin{proof}
By the tower property,
$$
\mathbb E_s[R^{(\ell)}]
=
\mathbb E_s\!\Big[\sum_{t=0}^{T-1}\ell_{X_t}(\omega_t)\Big]
=
\mathbb E_s\!\Big[\sum_{t=0}^{T-1}\mathbb E\!\big[\ell_{X_t}(\omega_t)\mid X_t\big]\Big]
=
\mathbb E_s\!\Big[\sum_{t=0}^{T-1}\psi_{X_t}\Big].
$$
The last expression is the accumulated node reward with reward vector $\psi$, so by
\eqref{eq:fundamental}--\eqref{eq:visits},
$
\mathbb E_s[R^{(\ell)}]=sN\psi.
$
Dividing by $\mathbb E_s[T]=sN\mathbf 1$ gives \eqref{eq:mr-afc-def}.

For the uniform-step representation,
$
\widehat{\mathbb E}_s\!\big[\ell_{X_U}(\omega_U)\big]
=
\frac{1}{\mathbb E_s[T]}\,
\mathbb E_s\!\Big[\sum_{t=0}^{T-1}\ell_{X_t}(\omega_t)\Big]
=
b_\ell(s),
$
which is \eqref{eq:mr-afc-uniform}.
\end{proof}

Proposition~\ref{prop:reward-afc-general} is the main reduction for the rest of this section: once a reward is
attached to the one-step simulator, the resulting reward-AFC is computed from the same fundamental matrix $N$. Applying
the construction to several reward systems in parallel yields a multi-reward profile on the same AMC.

\subsection{Valued local Top-\texorpdfstring{$\boldsymbol{k}$}{k} sets on each realized graph}

We now turn to the main case of interest: rewarding the entire local Top-$k$ set produced by the realized graph at each step, not just the single reported node.

For $i\in V$ and $\omega\in\Omega_i$, define the realized candidate set
\begin{equation}\label{eq:realized-topk-set}
\mathcal C_k(i,\omega)
:=
\begin{cases}
\mathrm{Top}^{\mathrm{loc}}_k(i;\Pi_i(\omega)), & \omega\notin\mathcal A_i,\\
\emptyset, & \omega\in\mathcal A_i.
\end{cases}
\end{equation}
Thus, invalid draws contribute no candidates, consistent with Definition~\ref{def:local-topk}.

Let $\gamma\in\mathbb R_+^V$ be a vector of node values. The step reward associated with the whole local Top-$k$ set is
$
\ell_i^{(k,\gamma)}(\omega)
:=
\sum_{v\in\mathcal C_k(i,\omega)}\gamma_v, i\in V.
$
Its expected one-step reward is $\psi_i^{(k,\gamma)}
:=
\mathbb E_i\!\big[\ell_i^{(k,\gamma)}(\omega)\big]
=
\mathbb E_i\!\Big[\sum_{v\in\mathcal C_k(i,\omega)}\gamma_v\Big].
$
The corresponding valued Top-$k$ reward-AFC is
$
b_{k,\gamma}(s)
:=
b_{\ell^{(k,\gamma)}}(s)
=
\frac{sN\psi^{(k,\gamma)}}{sN\mathbf 1}.
$
By \eqref{eq:mr-afc-uniform},
$
b_{k,\gamma}(s)
=
\widehat{\mathbb E}_s\!\Big[\sum_{v\in\mathcal C_k(X_U,\omega_U)}\gamma_v\Big].
$
So $b_{k,\gamma}(s)$ is the average total value of the entire local Top-$k$ candidate set seen at a length-biased
uniformly sampled pre-absorption step.

If $\gamma\equiv\mathbf 1$, then $b_{k,\mathbf 1}(s) = \widehat{\mathbb E}_s\!\big[\,|\mathcal C_k(X_U,\omega_U)|\,\big]$, the expected size of the current local Top-$k$ set on that random step.

This quantity is different from a node-reward summary based only on the visited state $X_t$. Here the reward is attached to the whole candidate set generated by the realized graph at step $t$. This is exactly the relevant object when the application cares about the current collection of betweenness candidates rather than only the single reported center.

\subsection{Pool-restricted valued Top-\texorpdfstring{$\boldsymbol{k}$}{k} rewards}

Some applications may restrict our candidate pool to $W\subseteq V$, for example an empirical union of local Top-$k$ sets: $W := \bigcup_{H\in\mathcal H_0}\ \bigcup_{i\in V}\ \mathrm{Top}^{\mathrm{loc}}_k(i;H)$, for some finite $\mathcal H_0\subset\mathcal H$.

Given such a candidate pool $W$ and node values $\gamma\in\mathbb R_+^V$, let $\ell_i^{(k,W,\gamma)}(\omega) := \sum_{v\in\mathcal C_k(i,\omega)\cap W}\gamma_v, i\in V$. Let $\psi_i^{(k,W,\gamma)} := \mathbb E_i\!\big[\ell_i^{(k,W,\gamma)}(\omega)\big]$. Then the pool-restricted reward-AFC is $b_{k,W,\gamma}(s) := b_{\ell^{(k,W,\gamma)}}(s) = \frac{sN\psi^{(k,W,\gamma)}}{sN\mathbf 1}$. Equivalently: 
\begin{equation}\label{eq:mr-pool-afc}
b_{k,W,\gamma}(s) = \widehat{\mathbb E}_s\!\Big[\sum_{v\in\mathcal C_k(X_U,\omega_U)\cap W}\gamma_v\Big].
\end{equation}

If $\gamma\equiv\mathbf 1$, this becomes $b_{k,W,\mathbf 1}(s) = \widehat{\mathbb E}_s\!\big[\,|\mathcal C_k(X_U,\omega_U)\cap W|\,\big]$, the expected number of members of the current local Top-$k$ set that fall in the pool $W$.

\subsection{Node-based and transition-based rewards as special cases}

The simulator-level reward formulation is not a separate competing model. Rather, it is a strict extension of the reward summaries already used earlier in the paper. This subsection is included for two reasons. First, it shows that the valued local Top-$k$ construction is backward-compatible with node-based and transition-based AFC. Second, it makes clear that all these summaries are computed on the same AMC through the same expected one-step reward vector and the same fundamental matrix.

Let $f:V\to\mathbb R_+$ and define $\ell_i^{(f)}(\omega):=f(i), i\in V$. Then $\psi_i=f(i)$, so \eqref{eq:mr-afc-def} gives
\begin{equation}\label{eq:mr-inner-product}
b_f(s)
=
b_{\ell^{(f)}}(s)
=
\frac{sNf}{sN\mathbf 1}
=
b(s)\,f
=
\sum_{v\in V}b_v(s)\,f_v.
\end{equation}
Thus the usual node-reward AFC is recovered by taking a reward that depends only on the currently visited node.

Let $\eta:S\times S\to\mathbb R_+$ with $\eta(\perp,\cdot)=0$, and define $\ell_i^{(\eta)}(\omega):=\eta\bigl(i,\mathrm{Next}_i(\omega)\bigr), i\in V$. Then $\sum_{t=0}^{T 1}\eta(X_t,X_{t+1}) = \sum_{t=0}^{T-1}\ell_{X_t}^{(\eta)}(\omega_t)$. Its expected one-step reward is $\psi_i^{(\eta)} = \mathbb E_i\!\big[\ell_i^{(\eta)}(\omega)\big] = \sum_{j\in V}Q_{ij}\,\eta(i,j)+P_{i\perp}\,\eta(i,\perp)$. Therefore
$
b_\eta(s)
:=
\frac{\mathbb E_s\!\big[\sum_{t=0}^{T-1}\eta(X_t,X_{t+1})\big]}{\mathbb E_s[T]}
=
\frac{sN\psi^{(\eta)}}{sN\mathbf 1}
=
b_{\psi^{(\eta)}}(s).
$
So transition rewards again reduce to the same reward-AFC formula.

A canonical example is the switching count $\eta_{\mathrm{sw}}(i,j):=\mathbf 1\{i\in V,\ j\in V,\ j\neq i\}$, and $\eta_{\mathrm{sw}}(i,\perp):=0$, which measures the average rate at which the reported center changes before absorption.

\section{Set-constrained Top-\texorpdfstring{$\boldsymbol{k}$}{k} selection and node-value summaries on the AMC}
\label{subsec:set-topk-phi-QN}

The previous section kept the AMC fixed and changed what was rewarded. In particular, the target pool $W$ entered only through the reward definition, while the one-step transition rule itself was unchanged. Here we take the complementary step: the pool constraint is enforced directly in the one-step update. This changes the induced kernel on $S=V\cup\{\perp\}$, but it does not enlarge the AMC state space in \eqref{eq:amc}.

Thus the distinction is as follows. In Section~\ref{subsec:multi-node-reward}, one scores the realized local Top-$k$ set produced at each step under a fixed kernel. In the present section, one changes the reported next state itself by requiring it to lie in a target pool whenever possible, and only then computes node-wise AFC and node-value summaries from the constrained kernel. This is useful when the structural constraint is part of the reporting rule rather than part of the reward.

We focus on three objects: the row-wise feasibility of the target constraint, the resulting node-wise occupancy profile, and the total occupancy mass on the target pool.

\subsection{Stepwise Top-\texorpdfstring{$\boldsymbol{k}$}{k} filtering and constrained selection}

Recall the realized candidate set $\mathcal C_k(i,\omega)$ from \eqref{eq:realized-topk-set}. Fix a deterministic
target pool $W\subseteq V$; in applications, $W$ is often a union of prescribed shapes such as $A\cup B$.

For each $i\in V$, define the constrained selector $\mathrm{Sel}_W(i,\omega)$ as follows:
if $\mathcal C_k(i,\omega)\cap W\neq\emptyset$, choose the highest-ranked element of
$\mathcal C_k(i,\omega)\cap W$ under the same deterministic ranking already used to form
$\mathrm{Top}^{\mathrm{loc}}_k$; otherwise set $\mathrm{Sel}_W(i,\omega)=\perp$.
Thus every nonabsorbing constrained update reports a node in $W$.

This induces a new AMC kernel
$
P^W=
\begin{bmatrix}
Q^W & r^W\\
0 & 1
\end{bmatrix}
$
on the same state space $S$, with $P^W_{ij}:=
\mathbb P_i\!\big(\mathrm{Sel}_W(i,\omega)=j\big), j\in V$,
and $P^W_{i\perp}:= 1-\sum_{j\in V}P^W_{ij}$. Under this hard filtering rule, absorption occurs either because the original draw is invalid or because the realized local Top-$k$ set has empty intersection with $W$. As in the previous sections, we assume the resulting kernel is absorbing; this is automatic, for example, under the same row-wise leak lower bound used earlier.

Define the one-step feasibility probability by
\begin{equation}\label{eq:settopk-xi}
\xi_i
:=
\mathbb P_i\!\big(\mathcal C_k(i,\omega)\cap W\neq\emptyset\big),
\qquad i\in V.
\end{equation}
Equivalently, $\xi_i=\sum_{j\in W}P^W_{ij}=1-P^W_{i\perp}$. So $\xi_i$ measures how often the target constraint can be satisfied directly from anchor $i$.

\subsection{Node-wise values and pool occupancy under the constrained kernel}

Let $N^W:=(I-Q^W)^{-1}$ be the fundamental matrix of the constrained kernel, and let $b^W(s):=\frac{sN^W}{sN^W\mathbf 1}$ be the corresponding node-wise AFC.

For any node-value vector $f\in\mathbb R_+^V$, define the constrained node-value summary by
$
b_f^W(s)
:=
\frac{sN^Wf}{sN^W\mathbf 1}
=
b^W(s)\,f.
$
Thus, once the kernel is replaced by $P^W$, node-wise value summaries follow from exactly the same linear-algebraic pipeline as in the unconstrained AMC. This differs from the pool-restricted reward $b_{k,W,\gamma}(s)$ in \eqref{eq:mr-pool-afc}, which leaves the kernel unchanged and modifies only the reward.

A basic global diagnostic is the AFC mass on the target pool:
\begin{equation}\label{eq:settopk-bPsi}
m_W(s)
:=
\sum_{v\in W} b_v^W(s)
=
b_{\mathbf 1_W}^W(s).
\end{equation}
This is the pre-absorption fraction of time spent in the target pool under the constrained kernel.
Under hard filtering, every post-initial nonabsorbing update lies in $W$, so off-pool mass can come only from initial
states outside $W$. In particular, if the initial distribution $s$ is supported on $W$, then $m_W(s)=1$.

When $m_W(s)>0$, the within-pool profile is obtained by renormalizing $\bar b_v^W(s):=\frac{b_v^W(s)}{m_W(s)}, v\in W$. This gives the occupancy distribution inside the target pool itself.

For ``motif''-structured targets \cite{milo_2002,benson_gleich_leskovec_2016}, we treat only the chosen primitives $A,B,\dots$ as admissible targets. Some of these ``motifs'' have been investigated in centrality studies, such as cliques \cite{vogiatzis2015integer,rysz2018finding} and induced stars \cite{vogiatzis2019identification}, among others \cite{rasti2022novel}. Mixed motifs formed by combining vertices across different primitives are not counted as separate targets unless they are already contained in at least one primitive. For example, if $A=\{1,2,3\}$ and $B=\{2,3,4\}$ are the designated triangles, then $\{1,3,4\}$ is not regarded as an additional target triangle unless it is itself one of the chosen primitives (i.e., a triangle). 

\subsection{Feasibility via a fixed fallback shape}

On steps where $\mathcal C_k(i,\omega)\cap W=\emptyset$, searching for alternative target witnesses may be expensive. A simpler implementation is to use a fixed fallback set $V_{\mathrm{fb}}\subseteq V$, assumed disjoint from $W$, whose induced subgraph contains an admissible shape in every realization (for example, because its internal edges are
deterministic). Fix a representative node $v_{\mathrm{fb}}\in V_{\mathrm{fb}}$.

Define the fallback selector by
{\small
\[
\mathrm{Sel}_{W,\mathrm{fb}}(i,\omega)
:=
\begin{cases}
\text{the highest-ranked element of }\mathcal C_k(i,\omega)\cap W,
& \mathcal C_k(i,\omega)\cap W\neq\emptyset,\\[2mm]
v_{\mathrm{fb}}, & \omega\notin\mathcal A_i,\ \mathcal C_k(i,\omega)\cap W=\emptyset,\\[1mm]
\perp, & \omega\in\mathcal A_i.
\end{cases}
\]}
This yields a second constrained kernel
$
P^{W,\mathrm{fb}}=
\begin{bmatrix}
Q^{W,\mathrm{fb}} & r^{W,\mathrm{fb}}\\
0 & 1
\end{bmatrix},
$
again on the same state space $S$. The feasibility probability $\xi_i$ from \eqref{eq:settopk-xi} is unchanged: it still records how often the target pool is reached directly. The difference is that filtered-empty but otherwise valid draws are now redirected to $v_{\mathrm{fb}}$ rather than absorbed.

Assuming this fallback kernel is absorbing, let $N^{W,\mathrm{fb}}:=(I-Q^{W,\mathrm{fb}})^{-1}$, and $b^{W,\mathrm{fb}}(s):=\frac{sN^{W,\mathrm{fb}}}{sN^{W,\mathrm{fb}}\mathbf 1}$. For any node-value vector $f\in\mathbb R_+^V$, define $b_f^{W,\mathrm{fb}}(s):=\frac{sN^{W,\mathrm{fb}}f}{sN^{W,\mathrm{fb}}\mathbf 1}$.


The corresponding pool mass is $m_W^{\mathrm{fb}}(s):=\sum_{v\in W} b_v^{W,\mathrm{fb}}(s)$, so that $1-m_W^{\mathrm{fb}}(s)$ is the pre-absorption fraction of time spent outside the target pool. This off-pool mass is driven mainly by fallback detours, together with any initial mass placed outside $W$.

One may keep fallback detours rare by connecting $V_{\mathrm{fb}}$ to $V\setminus V_{\mathrm{fb}}$ only through sparse or very low-probability links. The role of the fallback set is not to represent a target structure, but to maintain a well-defined continuation rule without enlarging the AMC state space.

\subsection{Optional censoring of fallback visits}

Sometimes the fallback states are used only as an implementation device, and one wishes to remove their contribution in the computations. At the path level this corresponds to deleting indices where $X_t\in V_{\mathrm{fb}}$ (e.g., a segment $3\!-\!v_{\mathrm{fb}}\!-\!4$ is read as $3\!-\!4$ after deletion). Note that this pathwise operation need not preserve the Markov property. At the level of AFC, however, it amounts to conditioning the uniform pre-absorption step on avoiding the fallback set.

Specifically, if $c_{\mathrm{fb}}(s):=\sum_{u\in V\setminus V_{\mathrm{fb}}} b_u^{W,\mathrm{fb}}(s)>0$, let $\widetilde b_v(s) := \frac{b_v^{W,\mathrm{fb}}(s)}{c_{\mathrm{fb}}(s)}$, and $v\in V\setminus V_{\mathrm{fb}}$. Then $\widetilde b_v(s) = \widehat{\mathbb P}_s^{W,\mathrm{fb}}\!\big(X_U=v \mid X_U\notin V_{\mathrm{fb}}\big), v\in V\setminus V_{\mathrm{fb}}$, where $U$ is the length-biased uniform pre-absorption step for the fallback kernel. Thus censoring removes the contribution of fallback detours at the occupancy level without changing the underlying AMC construction or the downstream linear-algebraic pipeline.

\begin{figure}[htbp]
    \centering
    \includegraphics[width=0.6\linewidth]{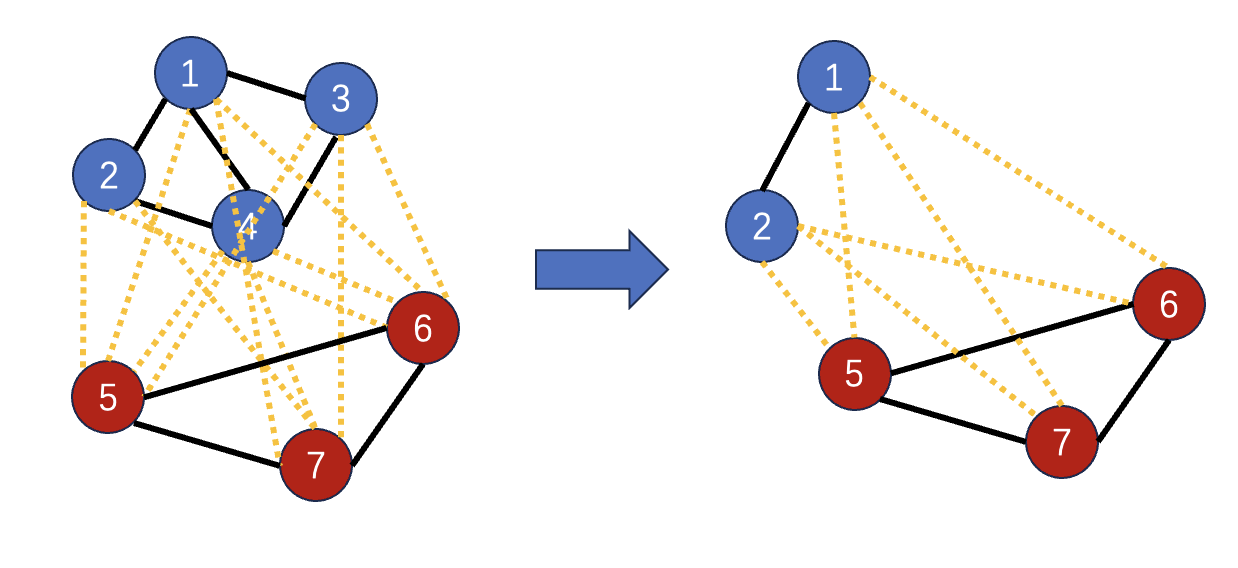}
    \caption{Illustrative fallback construction: primary nodes $1$--$4$ and a fallback triangle
    $V_{\mathrm{fb}}=\{5,6,7\}$. The figure uses full fallback connectivity only for visualization; in applications one
    may instead connect the fallback shape to the rest of the network through sparse or low-probability links so that
    fallback detours remain rare.}
    \label{fig:simple_set}
\end{figure}

Figure~\ref{fig:simple_set} illustrates this idea. The primary network consists of nodes $\{1,2,3,4\}$, while $\{5,6,7\}$ forms the fallback triangle. Assume that, at some realized step, the target pool cannot be reached through the current local Top-$k$ set. This could happen because, for example, the realized graph effectively leaves only a small component such as $\{1,2\}$. Then, the procedure will output the fallback representative $v_{\mathrm{fb}}\in\{5,6,7\}$ for that particular step.


\section{Constructing and estimating the absorbing Markov chain from the stochastic network}
\label{subsec:construct-amc}

The previous sections treated the absorbing Markov chain through its kernel $P$. To apply the framework, one must
construct that kernel from the underlying stochastic network model. This section makes that step explicit. We first
recall how the one-step simulator from Subsection~\ref{subsec:wellposed} induces each row of $P$, and then describe a
Monte Carlo estimator when the row probabilities are not available in closed form. The same construction also supports
the reward-aware and set-constrained variants developed earlier.

\subsection{One-step simulation and the induced row law}

Fix $i\in V$. As in Subsection~\ref{subsec:wellposed}, let $(\Omega_i,\mathcal F_i,\mathbb P_i),
\Pi_i:\Omega_i\to\mathcal H, \mathcal A_i $ denote the one-step simulator, realized-graph map, and effective absorption event, respectively. The next state is the tie-broken local center $c_{\mathrm{loc}}(i;\Pi_i(\omega))$ from \eqref{eq:local-center} on $\Omega_i\setminus\mathcal A_i$, and is $\perp$ on $\mathcal A_i$; equivalently, the row law is exactly the one defined in \eqref{eq:def-Pij}.

Algorithm~\ref{alg:sample-next} makes this construction operational. It is the basic primitive that turns one draw from the stochastic network model into one AMC step from anchor $i$. In particular, the exact row law of the AMC is the law of the output of Algorithm~\ref{alg:sample-next}, and the Monte Carlo estimator in Algorithm~\ref{alg:construct-amc} is obtained by repeating Algorithm~\ref{alg:sample-next} independently row by row.

\begin{algorithm}
  \caption{\textsc{SampleNext}$(i)$: primitive one-step simulator for the AMC row from anchor $i$}
  \label{alg:sample-next}
  \begin{algorithmic}[1]
    \REQUIRE Current local center $i\in V$.
    \ENSURE Next AMC state $X_{t+1}\in V\cup\{\perp\}$.
    \STATE Sample $\omega\sim\mathbb P_i$ on $(\Omega_i,\mathcal F_i)$.
    \STATE $H\gets \Pi_i(\omega)$.
    \IF{$\omega\in\mathcal A_i$}
      \STATE \textbf{return} $X_{t+1}\gets\perp$
      \COMMENT{this includes the earlier invalidity rule, e.g.\ $|K(i;H)|<k_{\min}$}
    \ENDIF
    \STATE \textbf{return} the tie-broken local center $c_{\mathrm{loc}}(i;H)$
    \COMMENT{using the fixed deterministic tie-breaker from the earlier definitions}
  \end{algorithmic}
\end{algorithm}

Algorithm~\ref{alg:sample-next} is not merely a procedural description: it is the mechanism that defines the transition row from $i$. Therefore the probabilities below are simply the output law of Algorithm~\ref{alg:sample-next}. For $j\in V$,
\begin{equation}\label{eq:sample-next-Pij}
P_{ij}:=\mathbb P_i\!\big(\textsc{SampleNext}(i)=j\big),
\end{equation}
and
\begin{equation}\label{eq:sample-next-Piabs}
P_{i\perp}:=\mathbb P_i\!\big(\textsc{SampleNext}(i)=\perp\big).
\end{equation}
These coincide with \eqref{eq:def-Pij}, and row normalization follows from
Lemma~\ref{lem:row-normalization}.

Let $\nu_i$ be the pushforward law on realized graphs from Definition~\ref{def:pushforward}, and let $a_i(\cdot)$ be the conditional absorption function from Proposition~\ref{prop:order-invariance}. Then the row probabilities are 
$
P_{ij}
=
\int_{\mathcal H}
\mathbf 1\{c_{\mathrm{loc}}(i;H)=j\}\,
\bigl(1-a_i(H)\bigr)\,d\nu_i(H),
j\in V,
$
and
$
P_{i\perp}
=
\int_{\mathcal H} a_i(H)\,d\nu_i(H).
$
Hence any two generation procedures inducing the same $\nu_i$ and $a_i(\cdot)$ yield the same AMC row, by Proposition~\ref{prop:order-invariance}.

\subsection{Monte Carlo estimation of the transition matrix}

When the row probabilities are not available in closed form, Algorithm~\ref{alg:sample-next} can be used as a sampling oracle for each transient row. Repeating it independently $M$ times for every anchor $i$ yields the empirical row estimator below. Algorithm~\ref{alg:construct-amc} summarizes the resulting row-wise Monte Carlo construction of the entire transition matrix.

For each $i\in V$, draw i.i.d.\ outputs
$
Z_i^{(1)},\dots,Z_i^{(M)}
\stackrel{d}{=}
\textsc{SampleNext}(i),
$
and define
$
\widehat P_{ij}
:=
\frac{1}{M}\sum_{m=1}^M \mathbf 1\{Z_i^{(m)}=j\},
j\in V\cup\{\perp\}.
$

\begin{algorithm}
  \caption{Row-wise Monte Carlo construction of the AMC transition matrix from repeated calls to \textsc{SampleNext}}
  \label{alg:construct-amc}
  \begin{algorithmic}[1]
    \REQUIRE Node set $V=\{1,\dots,n\}$; samples per row $M\in\mathbb N$.
    \ENSURE Estimated transition matrix $\widehat P$ on $S=V\cup\{\perp\}$.
    \FORALL{$i\in V$}
      \STATE Estimate $(\widehat P_{ij})_{j\in V\cup\{\perp\}}$ by empirical frequencies of \textsc{SampleNext}$(i)$.
    \ENDFOR
    \STATE Set $\widehat P_{\perp\perp}\gets 1$ and $\widehat P_{\perp j}\gets 0$ for $j\in V$.
  \end{algorithmic}
\end{algorithm}

Algorithm~\ref{alg:construct-amc} is the basic estimation routine used: once $\widehat P$ is available, AFC, reward-aware summaries, robust post-processing, and set-constrained variants all follow from the same downstream linear-algebraic computations.

For each fixed pair $(i,j)$, the strong law of large numbers gives $\widehat P_{ij}\xrightarrow{\mathrm{a.s.}} P_{ij}$ as $M\to\infty$.

We note that exact computation is generally intractable, hence the need for Monte Carlo estimation. Under independent edge failures, evaluating a single entry $P_{ij}$ amounts to integrating the indicator $\{c_{\mathrm{loc}}(i;H)=j\}$ over $2^{|E|}$ realizations, and closely related possible-world quantities are \#P-hard \cite{valiant1979complexity}. We therefore regard \eqref{eq:sample-next-Pij}--\eqref{eq:sample-next-Piabs} as quantities to be estimated rather than computed. 

In finite samples, one may observe $\widehat P_{i\perp}=0$ for some $i$, which can make
$(I-\widehat Q)$ ill-conditioned. Standard stabilizations include pseudocount (Laplace) smoothing or enforcing a small floor on $\widehat P_{i\perp}$ followed by row renormalization. After stabilization one forms $\widehat Q, 
\widehat N:=(I-\widehat Q)^{-1}, \widehat b(s):=\frac{s\widehat N}{s\widehat N\mathbf 1}$.

\subsection{Post-processing extensions on the estimated AMC}

Once the nominal kernel $\widehat P$ has been constructed by Algorithm~\ref{alg:construct-amc}, the extensions developed earlier in the paper can be applied without changing the basic estimation pipeline.

Starting from the nominal estimate $\widehat P$ produced by Algorithm~\ref{alg:construct-amc}, one may apply the row-wise uncertainty model of Subsection~\ref{subsec:robust-perturb-b}. For a continuation-value vector $v\in\mathbb R^{V\cup\{\perp\}}$, the row-wise robust inner problem is
\[
\mathcal T_i(v)
:=
\inf_{q_i\in\mathcal U_i^{\pm}}
\left[
\sum_{j\in V} q_{ij}v_j
+
\Bigl(1-\sum_{j\in V}q_{ij}\Bigr)v_\perp
\right],
\qquad i\in V.
\]
A scenario kernel $P^\star\in\mathcal U^{\pm}(\widehat P)$ is then selected by the specified robust rule (row-wise linear programs, a robust Bellman operator, or another criterion), and AFC is computed from $P^\star$ through the same formulas as before.

Once Algorithm~\ref{alg:construct-amc} has produced $\widehat P$ and hence $\widehat N$, any node-value summary follows immediately: for $f\in\mathbb R_+^V$, $\widehat b_f(s):=\frac{s\widehat N f}{s\widehat N\mathbf 1}$, which is the empirical version of \eqref{eq:mr-inner-product}. More generally, for a simulator-level reward family $\ell=(\ell_i)_{i\in V}$ from Proposition~\ref{prop:reward-afc-general}, estimate the expected one-step reward vector $\widehat\psi_i:=\frac{1}{M}\sum_{m=1}^M \ell_i(\omega_i^{(m)})$, using the same row samples, and then compute $\widehat b_\ell(s):=\frac{s\widehat N\widehat\psi}{s\widehat N\mathbf 1}$. This covers, in particular, the valued local Top-$k$ rewards and the pool-restricted rewards in \eqref{eq:mr-pool-afc}.

For the constrained selectors from Section~\ref{subsec:set-topk-phi-QN}, Algorithm~\ref{alg:construct-amc} is applied unchanged to the modified one-step rule, thereby yielding $\widehat P^W$ or $\widehat P^{W,\mathrm{fb}}$. Their corresponding AFC vectors are then $\widehat b^W(s)=\frac{s\widehat N^W}{s\widehat N^W\mathbf 1}$, and $\widehat b^{W,\mathrm{fb}}(s)=\frac{s\widehat N^{W,\mathrm{fb}}}{s\widehat N^{W,\mathrm{fb}}\mathbf 1}$. Using the same samples, estimate the feasibility probabilities in \eqref{eq:settopk-xi} by $\widehat\xi_i := \frac{1}{M}\sum_{m=1}^M \mathbf 1\{\mathcal C_k(i,\omega_i^{(m)})\cap W\neq\emptyset\}$, and record the fallback activation rate when relevant. The corresponding pool masses $\widehat m_W(s)$ and $\widehat m_W^{\mathrm{fb}}(s)$ are obtained by summing the relevant estimated AFC coordinates over $W$.

\subsection{Finite-sample accuracy: how many realizations are needed?}
\label{subsec:sample-size}

We now record conservative sample-size guarantees for two settings:
(i) estimating the one-shot law in the connected/anchor-free special case, and
(ii) estimating the full transition matrix (hence $b(s)$) in the general component-wise regime.
Errors are measured in $\|\cdot\|_\infty$ and $\|\cdot\|_1$
(total variation up to a factor $1/2$) \cite{david_asher_levin_y_peres_wilmer_propp_wilson_2017}.

\subsubsection{Connected/anchor-free special case: estimating the one-shot law \texorpdfstring{$\boldsymbol{p}$}{p}}

Recall the one-shot center law $p$ from \eqref{eq:oneshot-p}. Given i.i.d.\ samples
$Y^{(1)},\dots,Y^{(M)}\sim\mathcal L$, define $\widehat p_v:=\frac{1}{M}\sum_{m=1}^M \mathbf 1\{c(Y^{(m)})=v\}, v\in V$. For each fixed $v\in V$, the indicators are i.i.d.\ Bernoulli with mean $p_v$, so Hoeffding's inequality
\cite{hoeffding_1963} gives $\mathbb P\big(|\widehat p_v-p_v|\ge \varepsilon\big)
\le
2\exp(-2M\varepsilon^2).
$
Applying a union bound over $v\in V$ yields
$
\mathbb P\!\left(\|\widehat p-p\|_\infty\ge \varepsilon\right)
\le
2n\,\exp(-2M\varepsilon^2).
$
Equivalently, a sufficient condition for $\|\widehat p-p\|_\infty<\varepsilon$
with probability at least $1-\delta$ is
\begin{equation}\label{eq:phat-M}
M\ge \frac{1}{2\varepsilon^2}\log\!\Big(\frac{2n}{\delta}\Big).
\end{equation}

A conservative $\ell_1$ consequence is $\|\widehat p-p\|_1\le n\,\|\widehat p-p\|_\infty$. Hence, if
$
M\ge \frac{n^2}{2\varepsilon^2}\log\!\Big(\frac{2n}{\delta}\Big),
$
then $\|\widehat p-p\|_1<\varepsilon$ with probability at least $1-\delta$.

In the connected/anchor-free regime, \eqref{eq:phat-M} is a sufficient number of generated realizations to guarantee $\|p-\widehat p\|_\infty<\varepsilon$ with confidence $1-\delta$. Sharper multinomial concentration bounds are available
\cite{stephane_boucheron_gabor_lugosi_pascal_massart_2013}, but Hoeffding plus a union bound is often sufficient.

\subsubsection{General component-wise regime: estimating \texorpdfstring{$\boldsymbol{P}$}{P} and propagating error to \texorpdfstring{$\boldsymbol{b(s)}$}{b(s)}}

One call to Algorithm~\ref{alg:sample-next} returns a state in $V\cup\{\perp\}$ with row law
\eqref{eq:sample-next-Pij}--\eqref{eq:sample-next-Piabs}. With the empirical estimator from
Algorithm~\ref{alg:construct-amc}, Hoeffding's inequality gives, for any fixed $(i,j)$,
$
\mathbb P\big(|\widehat P_{ij}-P_{ij}|\ge\varepsilon\big)
\le
2e^{-2M\varepsilon^2}.
$
A union bound over all $n(n+1)$ pairs $(i,j)$ yields
$
\mathbb P\!\left(
\max_{i\in V}\max_{j\in V\cup\{\perp\}}
|\widehat P_{ij}-P_{ij}|
\ge \varepsilon
\right)
\le
2n(n+1)\exp(-2M\varepsilon^2).
$
Therefore, a sufficient condition for $\max_{i,j}|\widehat P_{ij}-P_{ij}|<\varepsilon$ with probability at least $1-\delta$ is
$
M\ge \frac{1}{2\varepsilon^2}\log\!\Big(\frac{2n(n+1)}{\delta}\Big).
$

Let $Q$ be the transient block of $P$ and $\widehat Q$ the transient block of $\widehat P$.
Assume a uniform hazard lower bound
$
r_i=P_{i\perp}\ge \underline r>0, i\in V,
$
so that $\|Q\|_\infty\le 1-\underline r$ and $N=(I-Q)^{-1}$ exists. Write
$
N:=(I-Q)^{-1}, \widehat N:=(I-\widehat Q)^{-1}, b(s):=\frac{sN}{sN\mathbf 1}, \widehat b(s):=\frac{s\widehat N}{s\widehat N\mathbf 1}.
$
To control $\|\widehat Q-Q\|_\infty$, note that
$
\|\widehat Q-Q\|_\infty
\le
n\,\max_{i\in V}\max_{j\in V}|\widehat P_{ij}-P_{ij}|.
$
Hence the entrywise guarantee
$
\max_{i\in V}\max_{j\in V\cup\{\perp\}}|\widehat P_{ij}-P_{ij}|<\frac{\varepsilon_Q}{n}
$
implies
$
\|\widehat Q-Q\|_\infty\le \varepsilon_Q.
$

Assume first that $\|\widehat Q-Q\|_\infty\le \frac{\underline r}{2}$. Then
$
\|\widehat Q\|_\infty\le 1-\frac{\underline r}{2},
$
so the Neumann series implies $\|N\|_\infty\le \frac{1}{\underline r}$, and 
$\|\widehat N\|_\infty\le \frac{2}{\underline r}$. Using the resolvent identity
$
\widehat N-N
=
(I-\widehat Q)^{-1}-(I-Q)^{-1}
=
\widehat N(\widehat Q-Q)N,
$
we obtain
\begin{equation}\label{eq:N-perturb}
\|\widehat N-N\|_\infty
\le
\frac{2}{\underline r^2}\,\|\widehat Q-Q\|_\infty.
\end{equation}

Now let $\mu:=sN$, and $\widehat\mu:=s\widehat N$. Then
$
b(s)=\frac{\mu}{\mu\mathbf 1}$,
and $\widehat b(s)=\frac{\widehat\mu}{\widehat\mu\mathbf 1}$. By \eqref{eq:N-perturb},
$
\|\widehat\mu-\mu\|_\infty
\le
\|\widehat N-N\|_\infty,
$ and $
\|\widehat\mu-\mu\|_1
\le
n\|\widehat\mu-\mu\|_\infty
\le
\frac{2n}{\underline r^2}\,\|\widehat Q-Q\|_\infty.
$
Also,
$
\mu\mathbf 1=\mathbb E_s[T]\ge 1,
$
since $T\ge 1$ almost surely when $X_0\in V$ almost surely.
If, in addition,
$
\|\widehat Q-Q\|_\infty\le \frac{\underline r^2}{8n},
$
then $\|\widehat\mu-\mu\|_1\le 1/4$, and therefore
$
\widehat\mu\mathbf 1
\ge
\mu\mathbf 1-\|\widehat\mu-\mu\|_1
\ge
\frac{3}{4}.
$

A direct normalization bound now gives
$$
\|\widehat b(s)-b(s)\|_1
\le
\frac{\|\widehat\mu-\mu\|_1}{\widehat\mu\mathbf 1}
+
\frac{|\widehat\mu\mathbf 1-\mu\mathbf 1|}{\widehat\mu\mathbf 1}
\le
\frac{2\|\widehat\mu-\mu\|_1}{\widehat\mu\mathbf 1}
\le
\frac{8}{3}\|\widehat\mu-\mu\|_1.
$$
Using the looser but cleaner constant $8$ produces
\begin{equation}\label{eq:b-perturb}
\|\widehat b(s)-b(s)\|_1
\le
\frac{8n}{\underline r^2}\,\|\widehat Q-Q\|_\infty.
\end{equation}
Hence, to guarantee $\|\widehat b(s)-b(s)\|_1\le \varepsilon_b$, it suffices to set
$\varepsilon_Q :=
\min\Bigl\{
\frac{\underline r}{2},
\frac{\underline r^2}{8n},
\frac{\varepsilon_b\,\underline r^2}{8n}
\Bigr\},
$ and require the per-row sample size to satisfy
$
M\ge \frac{n^2}{2\varepsilon_Q^2}\log\!\Big(\frac{2n(n+1)}{\delta}\Big).
$
Then, with probability at least $1-\delta$,
$
\max_{i,j}|\widehat P_{ij}-P_{ij}|<\frac{\varepsilon_Q}{n},
$
hence $\|\widehat Q-Q\|_\infty\le \varepsilon_Q$, and \eqref{eq:b-perturb} implies
$
\|\widehat b(s)-b(s)\|_1\le \varepsilon_b
$
for $s\in\Delta^{n-1}$.

We want to emphasize the gap between these worst-case guarantees and practice. At the experimental scale of Section~\ref{sec:numerical-experiments} ($n=100$), the bound above demands a per-row sample size many orders of magnitude beyond the $M=60$ that suffices empirically to stabilize the reported rankings. The bound is conservative in three ways: the union bound over all $n(n+1)$ entries, the $\ell_\infty$-to-$\ell_1$ conversion factor $n$, and the worst-case resolvent constant $1/\underline r^{2}$. It should therefore be read as an ex-post validation device rather than an ex-ante sample budget. In practice we recommend nonparametric bootstrap confidence intervals on $\widehat b(s)$, obtained by resampling the $M$ row draws with replacement, together with the rank-stability check of rerunning Algorithm~\ref{alg:construct-amc} at $M$ and $2M$ and verifying that the Top-$k$ set is unchanged.

The finite-sample bounds above are stated for the raw empirical estimator $\widehat P$. If one applies Laplace smoothing or a floor on $\widehat P_{i\perp}$ before forming $\widehat N$, the same perturbation argument still applies after accounting for the deterministic effect of the stabilization step. This effect is easy to quantify: imposing a floor of size $\eta$ on $\widehat P_{i\perp}$ followed by row renormalization changes each transient row by at most $\eta$ in $\ell_1$, so $\|\widetilde Q-\widehat Q\|_\infty\le\eta$ and, by \eqref{eq:b-perturb}, the stabilization bias in AFC is at most $8n\eta/\underline r^{2}$ in $\ell_1$; choosing $\eta$ an order of magnitude below the Monte Carlo error renders it negligible.

\section{Numerical Experiments}
\label{sec:numerical-experiments}

We are ready to evaluate our computational pipeline. Each realized working graph induces a betweenness-based local center and local Top-$k$ candidate set; the resulting one-step rule is compressed into an AMC on $S=V\cup\{\perp\}$ with transition matrix \eqref{eq:amc}. Given an estimated kernel, we compute AFC from the fundamental matrix via \eqref{eq:fundamental}--\eqref{eq:afc-def}. The same AMC estimates are subsequently reused for robust sensitivity analysis and reward/structure-aware
summaries (multi-reward and set-based Top-$k$ developments).

\subsection{Numerical test in Erd\H{o}s--R\'enyi and Watts--Strogatz}

We draw a base topology $G_0=(V,E)$ from two standard random graph models with $n=|V|=100$. For Erd\H{o}s--R\'enyi (ER) \cite{erdos_1988}, we sample $G_0\sim G(n,p)$ with $p=0.08$. For Watts--Strogatz (WS) \cite{watts_strogatz_1998}, we use ring degree $6$ and rewiring probability $0.10$.

At each step, we generate a realized working graph $H$ by retaining each base edge independently with probability $p_{\mathrm{on}}=0.85$. All computations use unweighted shortest paths, so betweenness is computed on the realized unweighted graph. One-step updates follow Algorithm~\ref{alg:sample-next} with a deterministic tie-breaker, and terminate either by exogenous stopping (absorption coin $\alpha=0.15$ per step) or by invalid continuation when the anchor’s connected component has a smaller size than $k_{\min}=5$.

For the matrix-based AMC pipeline, we estimate $\widehat P$ row-wise using Algorithm~\ref{alg:construct-amc}, with $M$ independent calls to the one-step routine per transient state (baseline $M=60$). The initial distribution is uniform, $s=\tfrac{1}{n}\mathbf 1$. To stabilize inversion of $I-\widehat Q$ at finite $M$, we impose a small positive floor on $\widehat P_{i\perp}$ (followed by renormalization) whenever sampling yields $\widehat P_{i\perp}=0$. 

\subsubsection{Baseline AFC under topology uncertainty}
\label{subsec:exp1-baseline}

For each ER/WS base topology, we estimate the AMC kernel $\widehat P$ via Algorithm~\ref{alg:construct-amc} under the within-step sampling model and parameters $(p_{\mathrm{on}},\alpha,k_{\min})$, then compute $\widehat b(s)$ using \eqref{eq:fundamental}--\eqref{eq:afc-def}. We report the Top-$5$ nodes under $\widehat b(s)$ and plot
the Top-$10$ values $\widehat b_v(s)$.

\begin{figure}[htbp]
    \centering
    \includegraphics[width=0.49\linewidth]{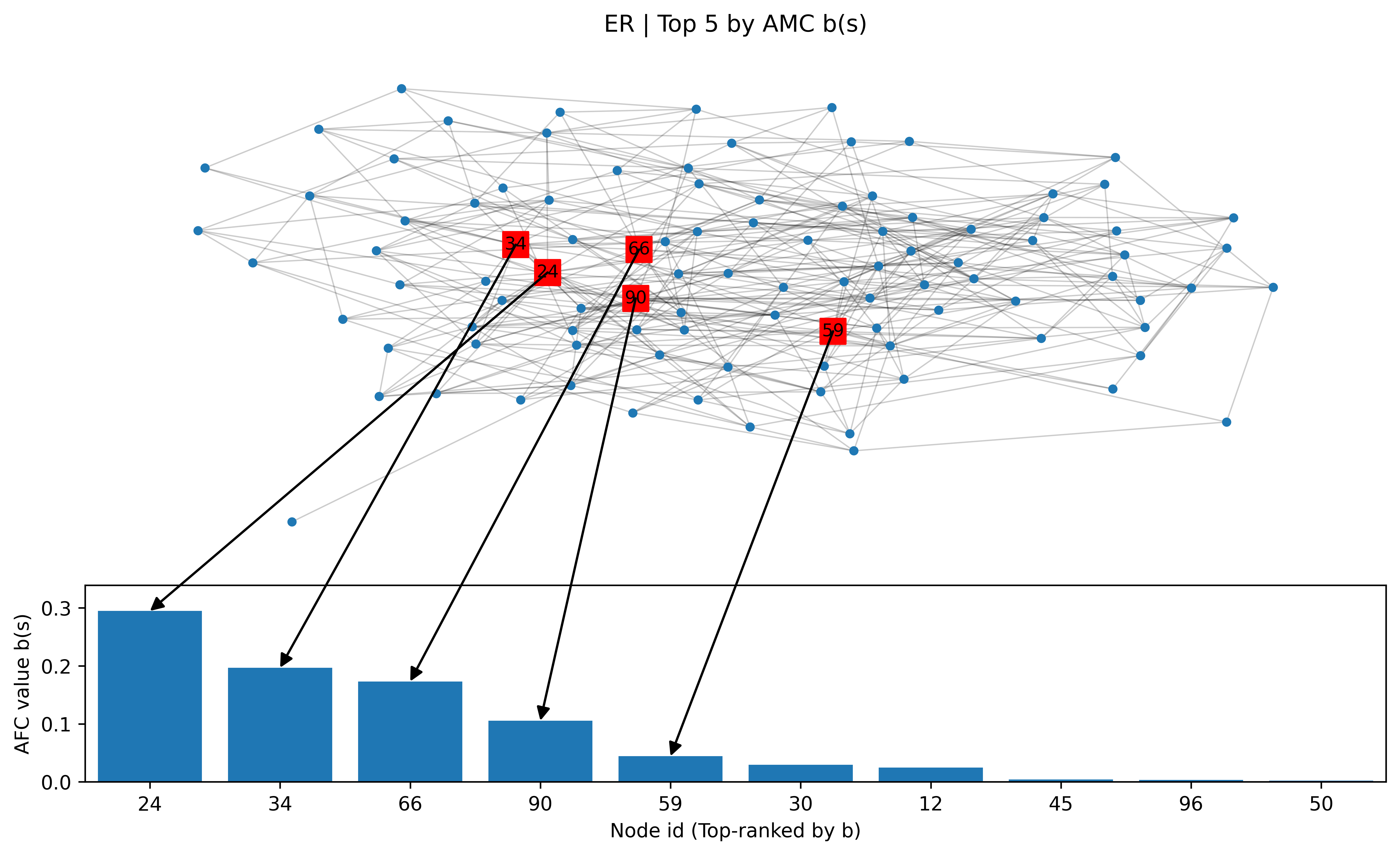}
    \includegraphics[width=0.49\linewidth]{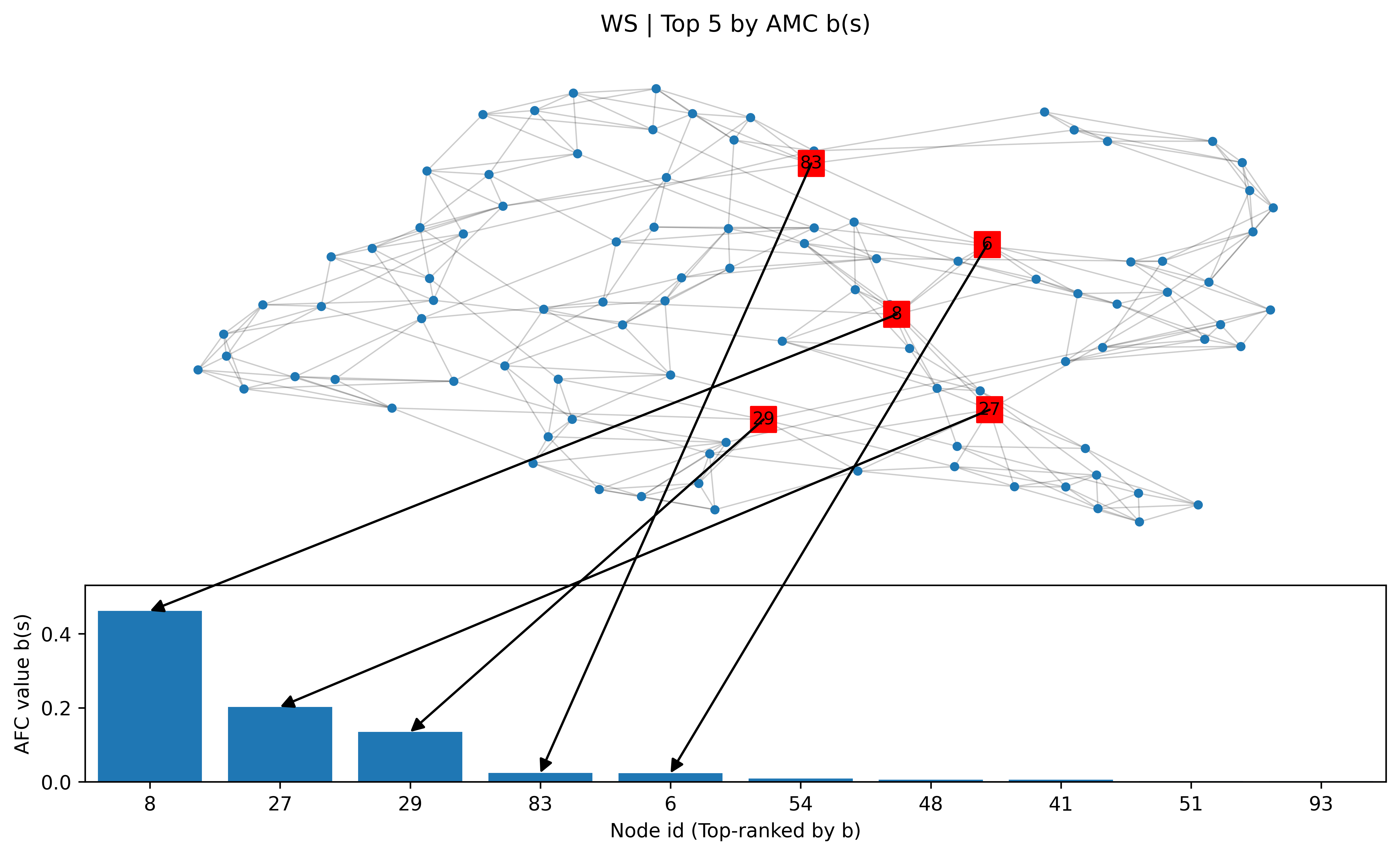}
    \caption{ Network plot on the base topology highlighting the Top-$5$ nodes under $\widehat b(s)$ in ER and WS, together with a bar chart of the Top-$10$ AFC values.}
    \label{fig:test1}
\end{figure}

Figure~\ref{fig:test1} reports the Top-$10$ nodes under $\widehat b(s)$ for the ER and WS networks. In the ER case, the top four nodes exhibit markedly larger AFC scores than the remaining nodes, while in the WS case a similar separation is observed among the top three nodes. Overall, these results indicate that the multi-step center dynamics concentrates occupancy on a small subset of nodes, even as the realized working graph varies from step to step and termination occurs via absorption.

\begin{figure}[htbp]
    \centering
    \includegraphics[width=0.99\linewidth]{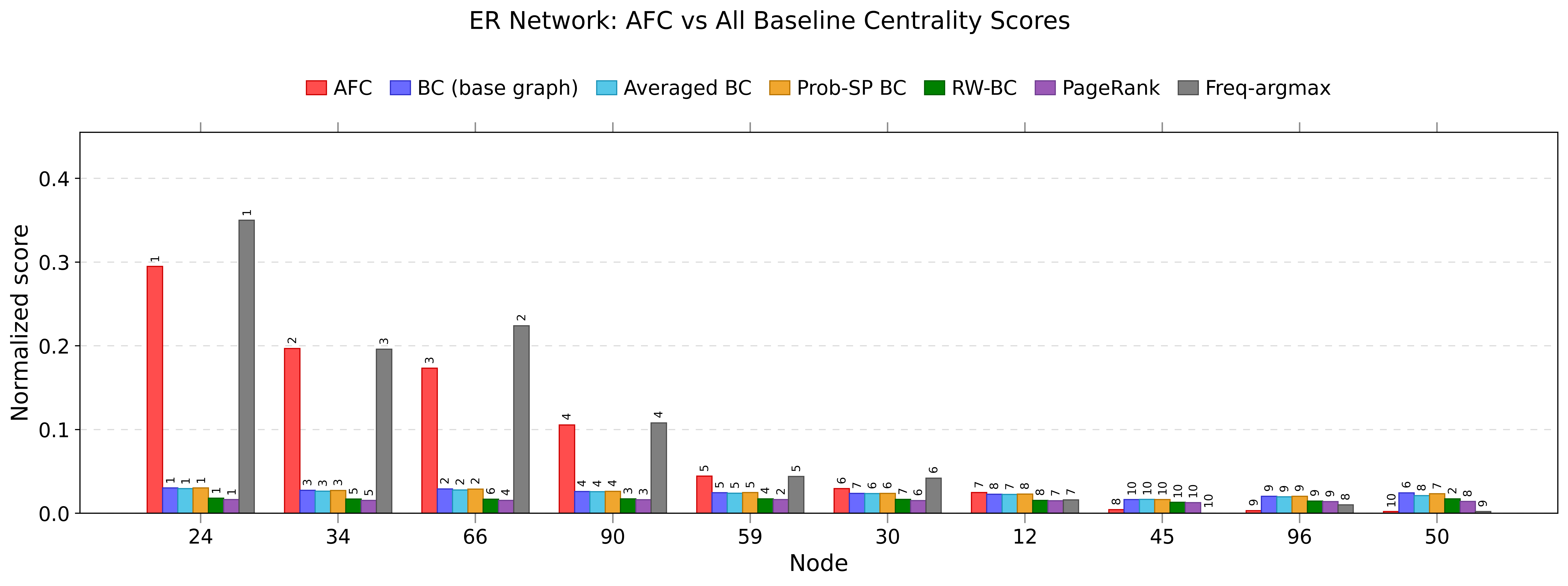}
    \includegraphics[width=0.99\linewidth]{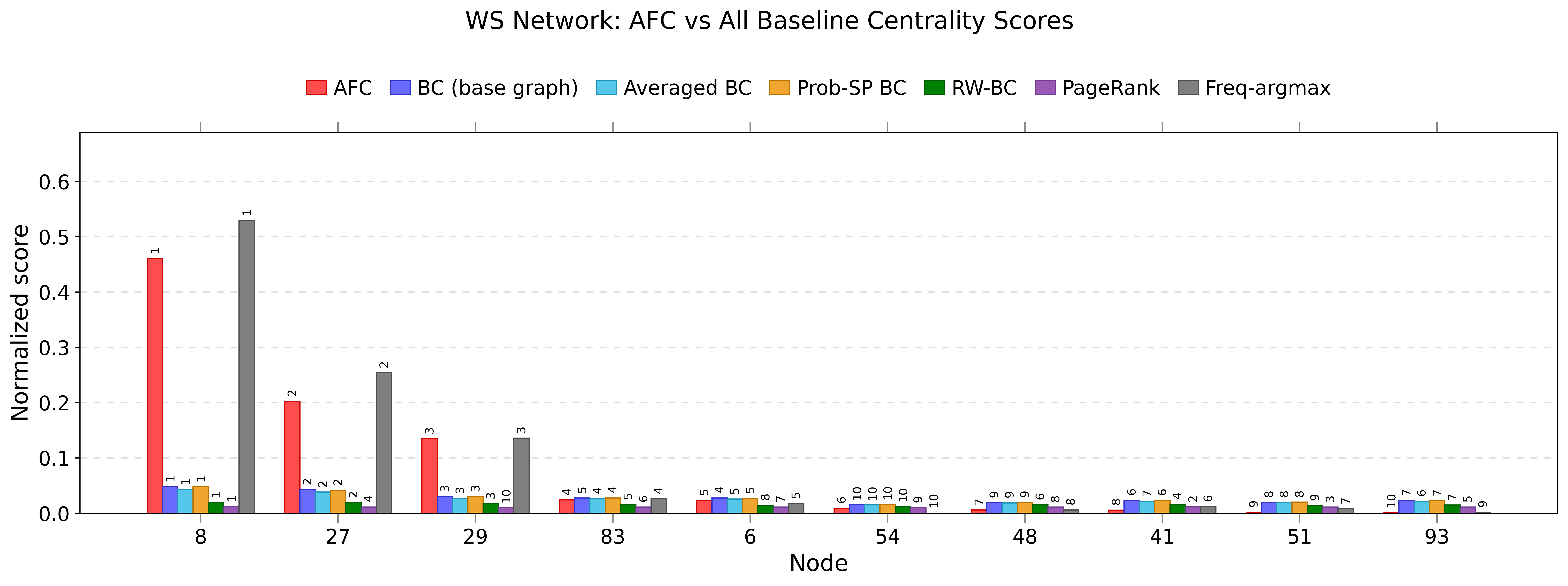}
    \caption{The left and right panels show results for the ER and WS networks, respectively. For each network, the x-axis lists the top ten nodes ranked by AFC. Bars show normalized scores across seven centrality measures: AFC, base-graph betweenness, averaged betweenness over sampled realizations, probabilistic shortest-path betweenness, random-walk betweenness, PageRank, and frequency-of-argmax centrality. Numbers above bars indicate each node's rank under that metric. Scores are normalized within each method to facilitate comparison across centrality definitions.}
    \label{fig:baseline-come}
\end{figure}

Figure~\ref{fig:baseline-come} compares AFC with several baseline centralities on the ER and WS networks. AFC provides much stronger discrimination among important nodes. In the ER network, the top five AFC nodes account for $81.5\%$ of the total normalized score, while the corresponding figures are $13.7\%$ for BC, $13.4\%$ for Averaged BC, $13.7\%$ for Prob-SP BC, $8.7\%$ for RW-BC, and $7.9\%$ for PageRank. A similar pattern holds in the WS network, where the top five AFC nodes capture $84.6\%$ of the score against much flatter baseline distributions. This suggests that BC, Averaged BC, Prob-SP BC, RW-BC, and PageRank have limited discriminative power after normalization. Although Freq-argmax also concentrates mass (92.2\% in ER and 96.4\% in WS), it is a winner-takes-all measure that ignores nodes consistently selected beyond first place. AFC instead accumulates occupation behavior across uncertain realizations, yielding a sharper and more informative ranking of structurally important nodes.

\subsubsection{Robust sensitivity under row-wise kernel perturbations} 
\label{subsec:exp2-robust}

We also evaluate how sensitive the AFC profile $b(s)$ is to finite-sample uncertainty in the estimated AMC
kernel $\widehat P$. Starting from the nominal estimate $\widehat P^{0}$ from
Experiment~\ref{subsec:exp1-baseline}, we perturb the transient block $\widehat Q^{0}$ row-wise and
renormalize with $P_{i\perp}=1-\sum_{j\in V}Q_{ij}$, enforcing $P_{i\perp}\ge r_{\min}$. We use
multiplicative perturbations with relative radius $\delta_{\mathrm{rel}}=0.50$ (clipped to $[0,1]$)
and set $r_{\min}=0.05$. Discrepancies are measured by KL divergence and by $W_1$, where the ground
metric on $V$ is shortest-path distance on the base topology.

For KL we sample $N_{\mathrm{KL}}=100$ admissible kernels and select the maximizer; for $W_1$ we sample $N_{\mathrm{W1}}=100$ and select the maximizer. The resulting AFC vectors are denoted $b^{\mathrm{KL}}(s)$ and $b^{\mathrm{W1}}(s)$. 

\begin{figure}[htbp]
    \centering
    \includegraphics[width=0.49\linewidth]{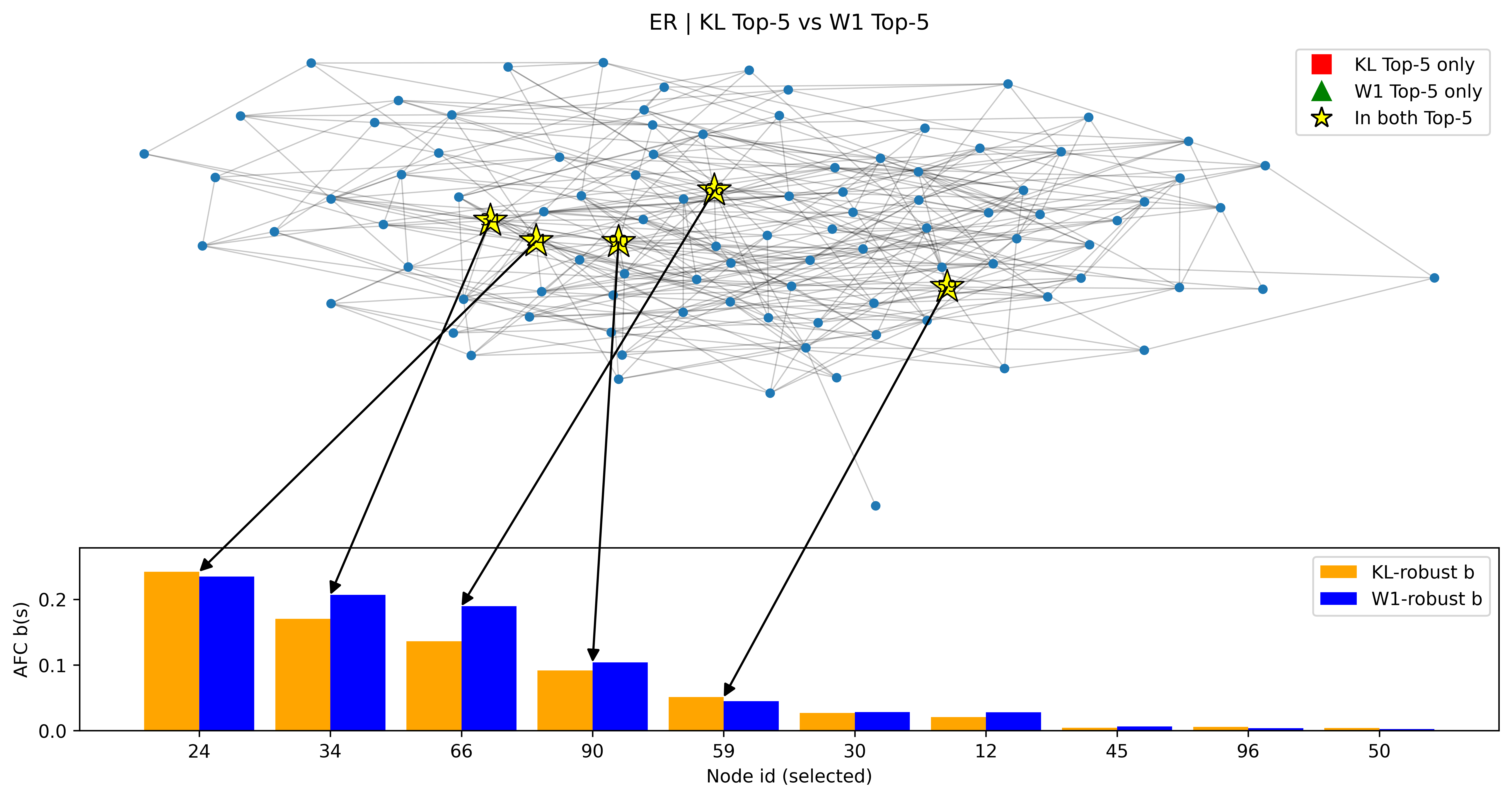}
    \includegraphics[width=0.49\linewidth]{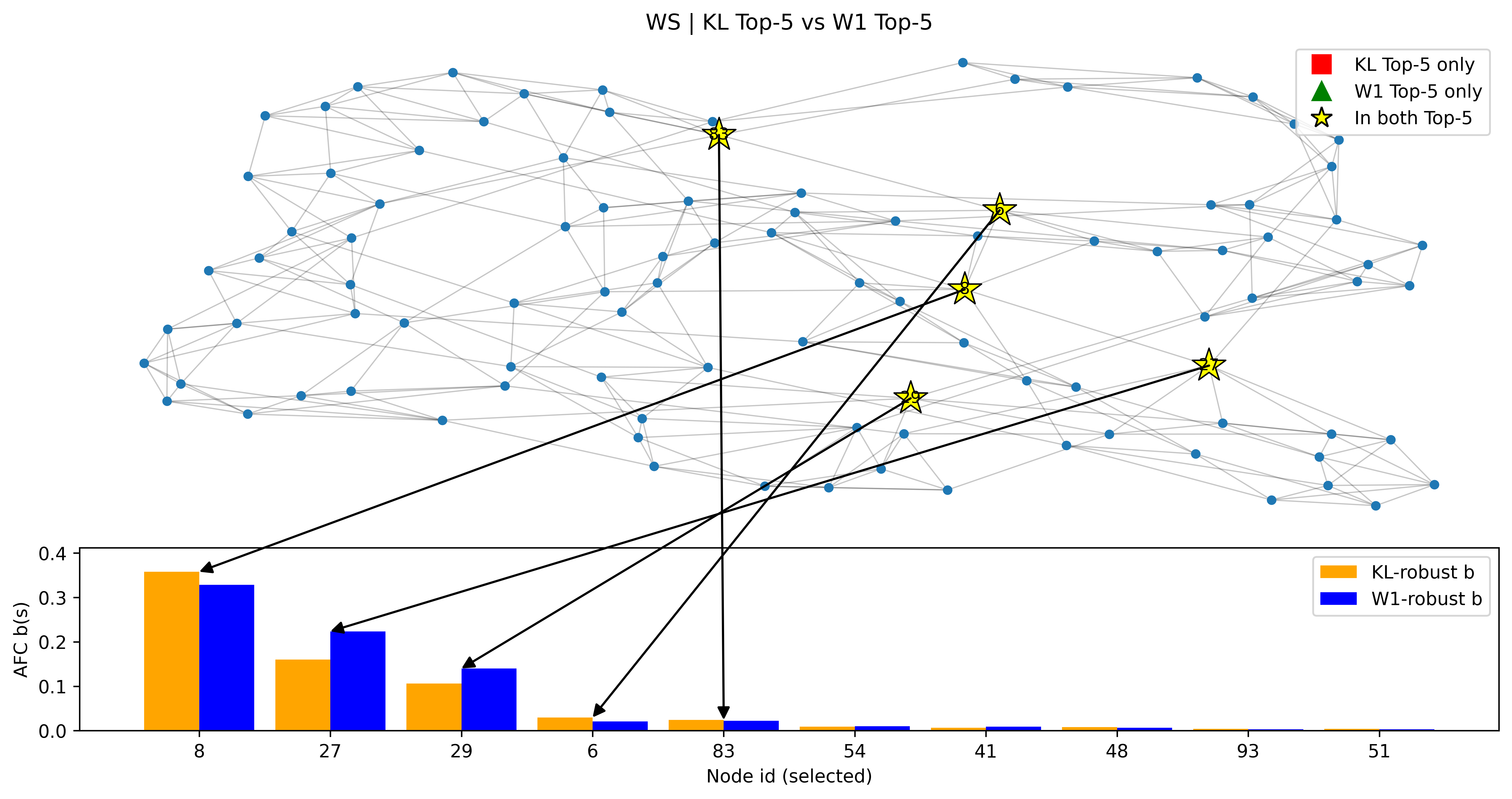}
    \caption{For each model, we mark the Top-(5) nodes under $b^{\mathrm{KL}}(s)$ and $b^{\mathrm{W1}}(s)$ on the network (with overlaps indicated) and plot a Top-(10) bar chart comparing $b^{\mathrm{KL}}_v(s)$ vs.$b^{\mathrm{W1}}_v(s)$ on the union of the two Top-(20) sets. This highlights that perturbations can change both magnitudes and rankings because AFC depends on $(I-Q)^{-1}$, not just one-step marginals.}
    \label{fig:test2}
\end{figure}

Figure~\ref{fig:test2} summarizes the robust search. For both ER and WS, the KL- and $W_1$-based
maximizers yield the same Top-$5$ node set, though with different AFC values. Relative to the baseline
in Figure~\ref{fig:test1}, WS preserves membership among the top three nodes but reorders ranks $4$--$5$:
the Top-$5$ WS nodes change from $(83,6,54,48,41)$ (Baseline) to $(6,83,54,41,48)$, swapping $83$ with $6$ and
exchanging $48$ with $41$.

\subsubsection{Multi-reward evaluation}
\label{subsec:exp4-multireward}

Next, we illustrate the reward-aware AFC from Section~\ref{subsec:multi-node-reward}. Given
$\widehat P$, all reward quantities are computed by post-processing the same AMC (no state-space
change). Let $h_1,h_2,h_3$ be the top-$3$ degree nodes in $G_0$, with rewards
$(R_1,R_2,R_3)=(10,10,10)$. Define reward function
$
f(v)=\max_{\ell\in\{1,2,3\}} R_\ell\,\beta^{d_{G_0}(v,h_\ell)}$ with $\beta=0.60,$ setting $f(v)=0$ if $v$ is unreachable from all hubs. On the same AMC we compute $b_f(s)=b(s)^\top f$ and two transition-based reward-AFC scalars: switching intensity (strict center changes) and improvement (nonnegative one-step increases in $f$) via the reduction in Section~\ref{subsec:multi-node-reward}.

\begin{figure}[htbp]
    \centering
    \includegraphics[width=0.49\linewidth]{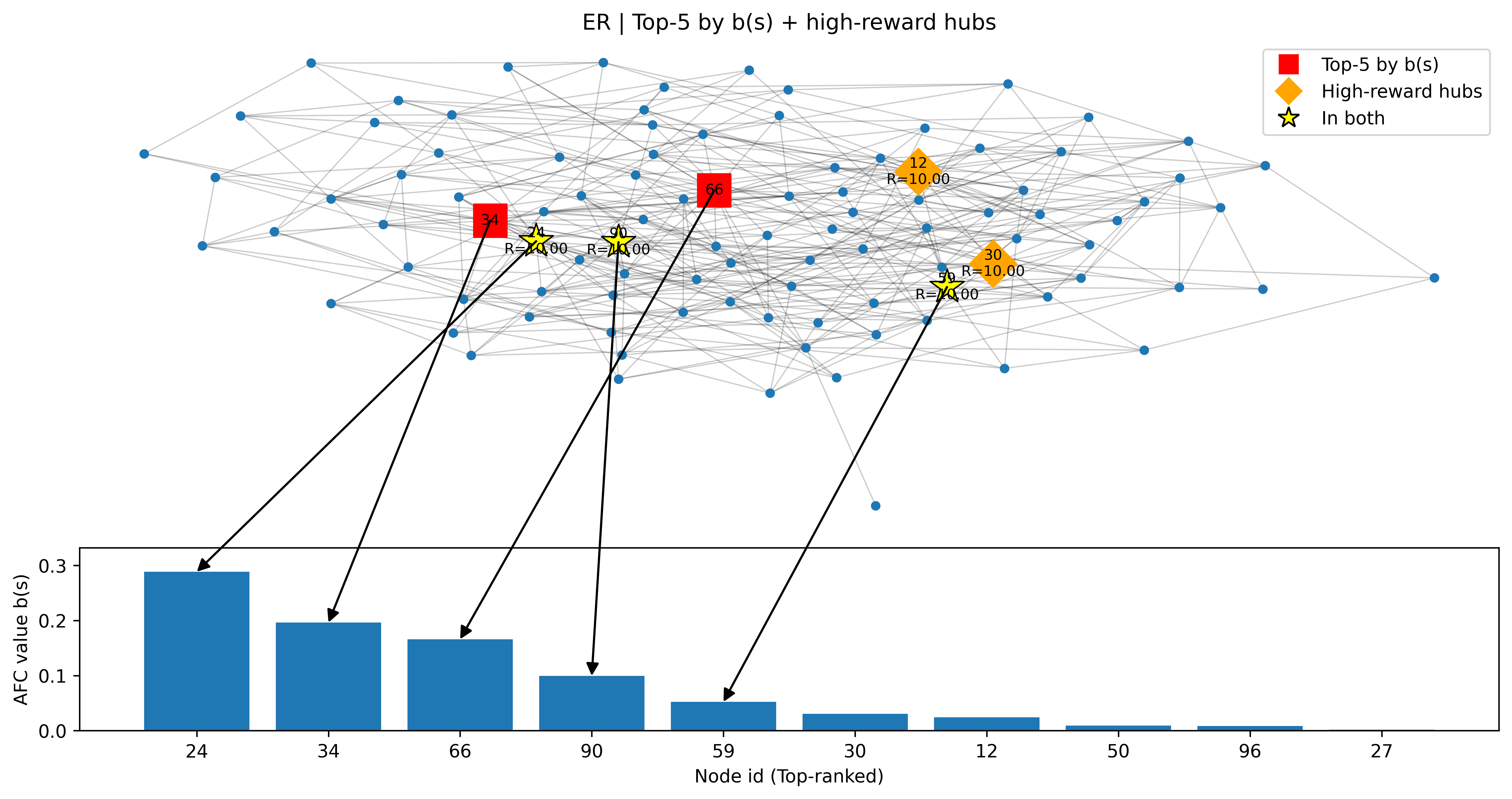}
    \includegraphics[width=0.49\linewidth]{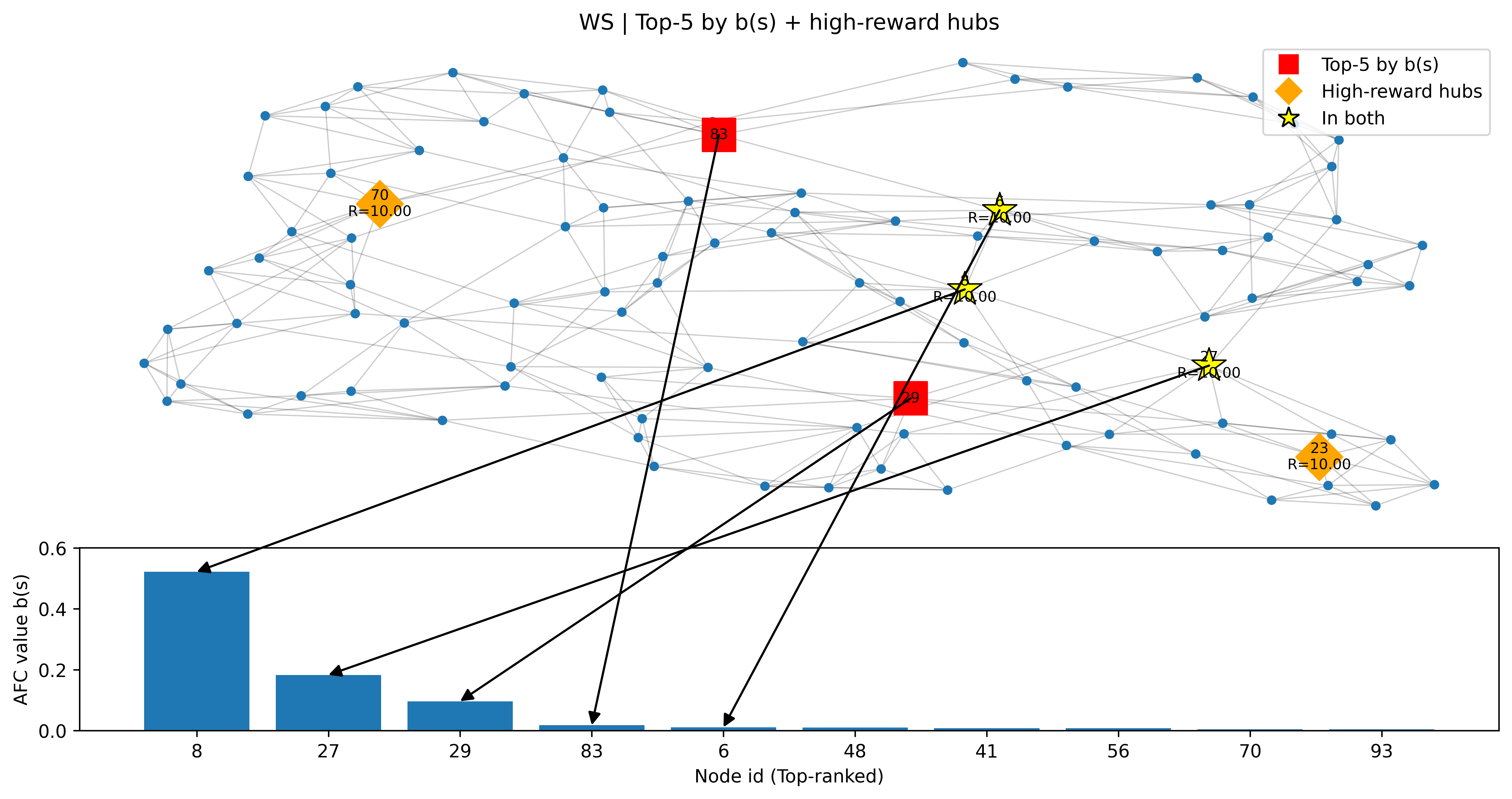}
    \caption{For each model, we report the Top-$5$ nodes under $\widehat b(s)$, mark the reward hubs, and report $b_f(s)$ plus the switching-intensity and improvement reward-AFC scalars.}
    \label{fig:test3}
\end{figure}

Figure~\ref{fig:test3} reports the results. Under our degree-based construction, the high-reward hub nodes are $\{8,70,6,23,27\}$. In the ER graph, the nodes ranked $8$th-$10$th are $\{50,96,27\}$ (vs. $\{45,96,50\}$ in the baseline network). In WS, the high-reward hub nodes are $\{8,27,29,83,6\}$. AFC mass concentrates on the top three nodes ($8,27,29$) with uniformly small values thereafter, contrasting with the baseline where nodes $83$ and $6$ remain comparatively prominent.

\subsubsection{Structure-constrained selection via \texorpdfstring{$\boldsymbol{3}$}{3}-clique target pools}
\label{subsec:exp3-structure}

We impose a motif constraint by restricting the reported next center to a target pool $W\subseteq V$
when feasible, implementing the set-based Top-$k$ filtering of Section~\ref{subsec:set-topk-phi-QN}. This modifies the induced AMC kernel while preserving the state space in \eqref{eq:amc}. We enumerate all $3$-cliques in $G_0$, score each clique by the sum of its node degrees, and select the top $S=8$ cliques (ties lexicographically). Let $W=\bigcup_{\ell=1}^{S}K_\ell$. On each non-absorbing step, after computing the local Top-$k$ set, we output the highest-ranked node in $\mathrm{Top}\text{-}k\cap W$ under the same deterministic ordering; if $\mathrm{Top}\text{-}k\cap W=\emptyset$, we fall back to a fixed $v_{\mathrm{fb}}\in W$ (the smallest id in the first selected clique). Exogenous stopping and small-component invalidity checks remain unchanged. 

\begin{figure}[htbp]
    \centering
    \includegraphics[width=0.49\linewidth]{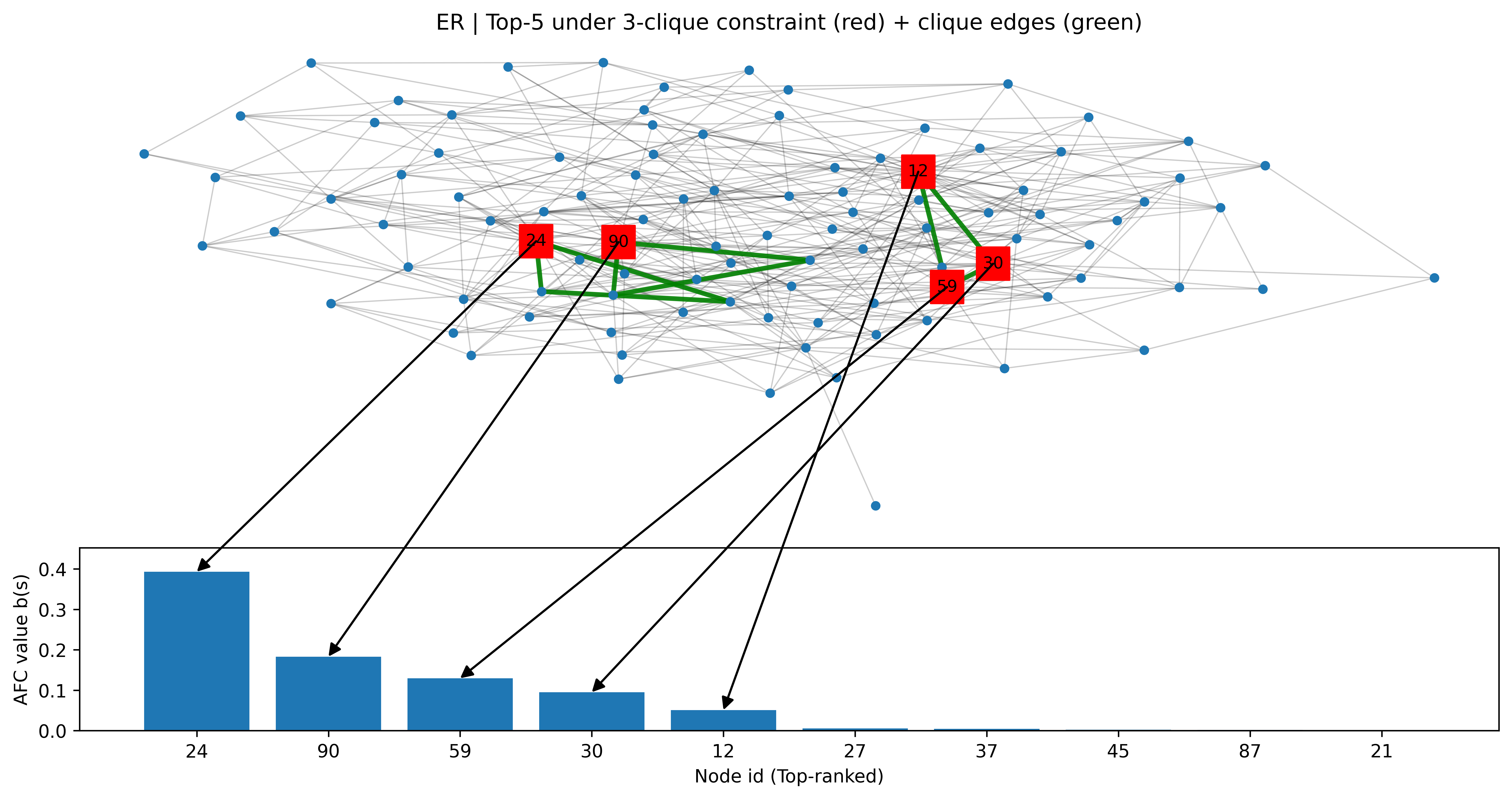}
    \includegraphics[width=0.49\linewidth]{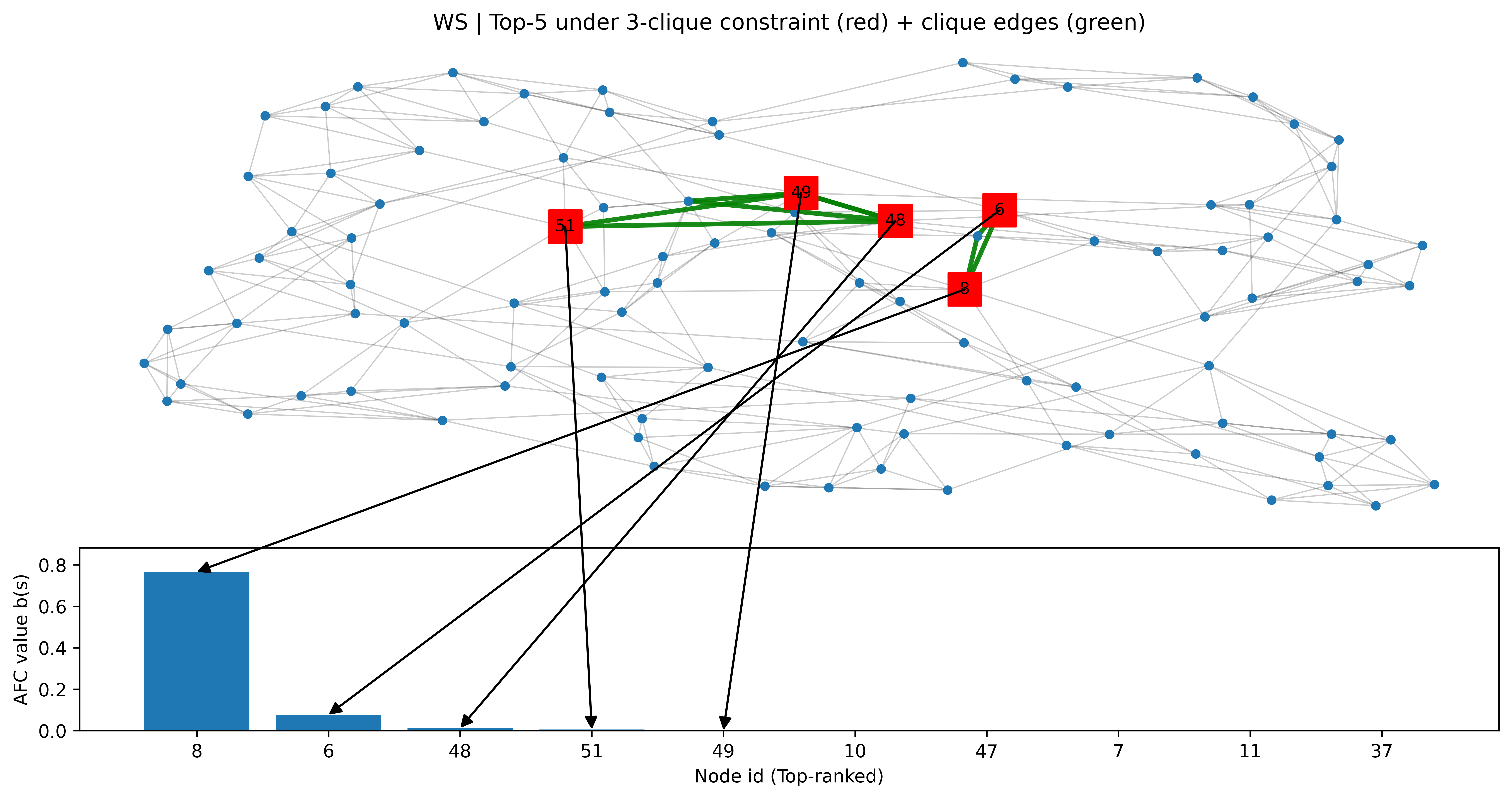}
    \caption{Top-$5$ nodes under the structure-constrained AFC profile, visualize them on the base network with the relevant motif edges emphasized, and plot the Top-$10$ AFC values. This highlights how motif feasibility and fallback detours can concentrate occupancy on $W$ while preserving \eqref{eq:amc}.}
    \label{fig:test4}
\end{figure}

Figure~\ref{fig:test4} shows that the induced rankings differ from the other tests, while the reported centers satisfy the imposed $3$-clique structure.

\subsection{Les Mis\'erables co-occurrence network}
\label{subsec:lesmis}

We apply the AMC--AFC pipeline to the Les Mis\'erables co-occurrence benchmark \cite{knuth1993stanfordgraphbase}, a well-known weighted network on $|V|=77$ vertices and $|E|=254$ edges. The integer weights satisfy $w_e^{0}\in[1,31]$, so we set $w_{\max}=31$. To introduce
stochastic heterogeneity without changing $E$, we resample edge weights at every call to $\texttt{SampleNext}(\cdot)$: conditional on $w_e^{0}$, we draw
$x_e\sim\mathcal N(\mu_e,\sigma_e^2)$ with
$
\mu_e=w_e^{0}+\rho_{\mu}(w_{\max}-w_e^{0}), 
\sigma_e=\rho_{\sigma}(w_{\max}-w_e^{0}),
$
and set $\tilde w_e=\Pi_{[w_e^{0},w_{\max}]}\!\big(\mathrm{round}(x_e)\big)$, i.e., rounded and clipped to $[w_e^{0},w_{\max}]$. 

Because the network is connected, we enforce anchor dependence by computing betweenness only on the $r$-hop neighborhood $S_r(i)$ of the current state $i$ and declaring absorption when $|S_r(i)|<k_{\min}$. At each non-absorbing step we resample $\{\tilde w_e\}$, compute weighted betweenness on the induced subgraph on $S_r(i)$ (using $\tilde w_e$ as edge lengths), and select the deterministic maximizer. We stop exogenously with probability $\alpha$ per step. 

The resulting AMC kernel $P^{0}$ has no closed form and is estimated row-wise via Monte Carlo, yielding $\widehat P^{0}$ from $M$ independent calls to $\texttt{SampleNext}(i)$ per state. We then compute AFC $b(s)$ from the fundamental matrix under a uniform initial distribution $s$, using the same stabilization (absorption floor and renormalization) as in the random-graph tests.

\begin{figure}[htbp]
    \centering
    \includegraphics[width=0.49\linewidth]{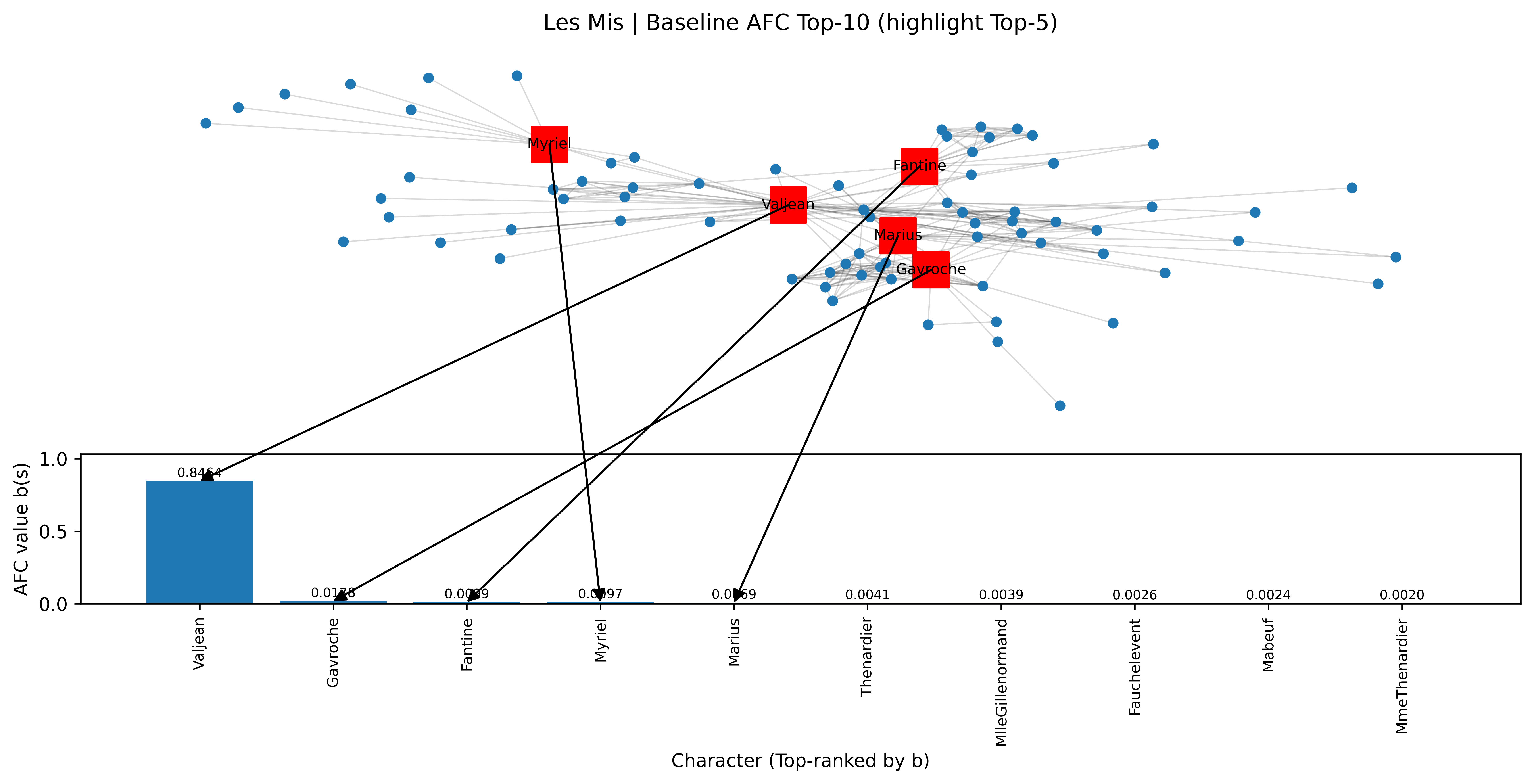}
    \includegraphics[width=0.49\linewidth]{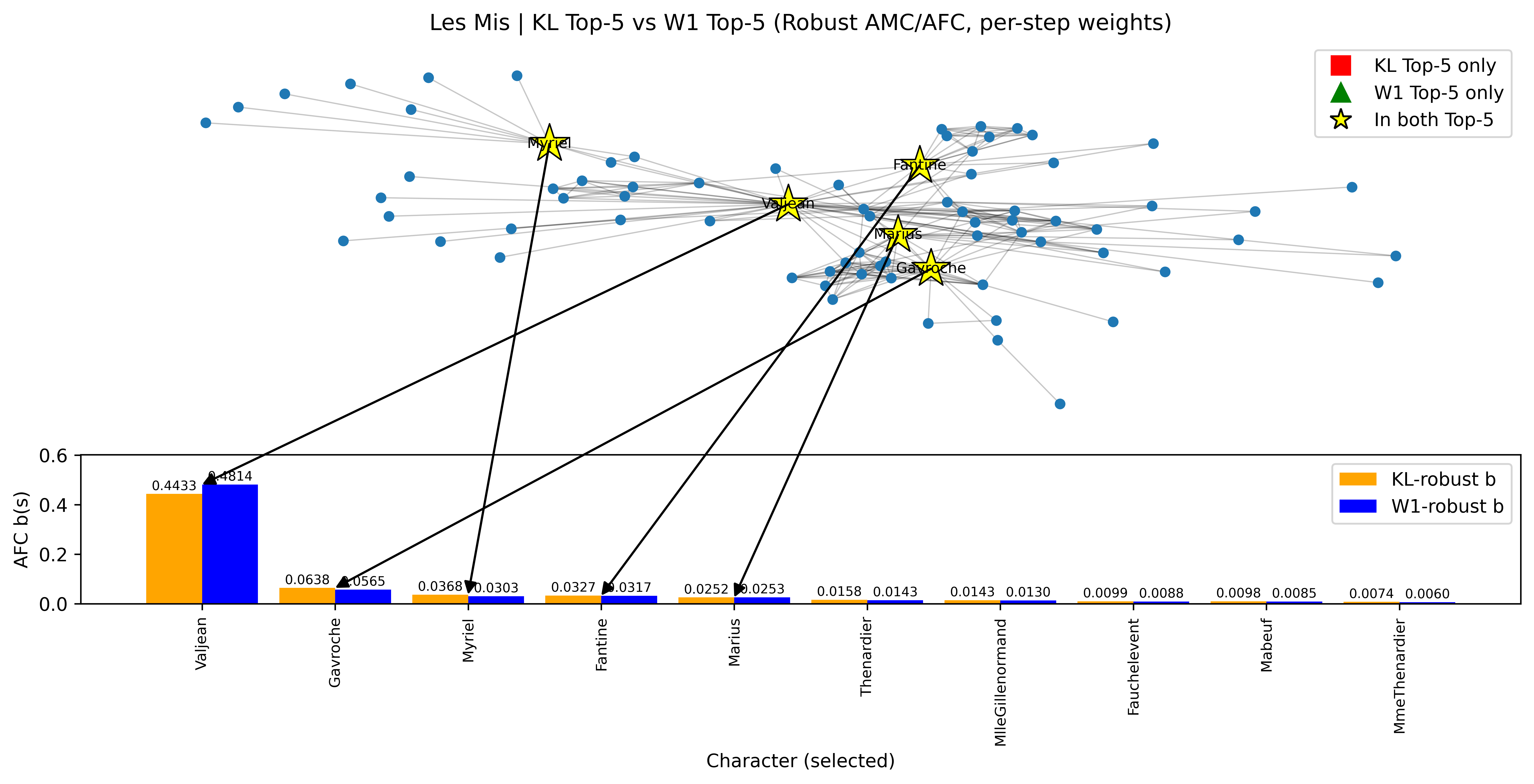}
     \includegraphics[width=0.49\linewidth]{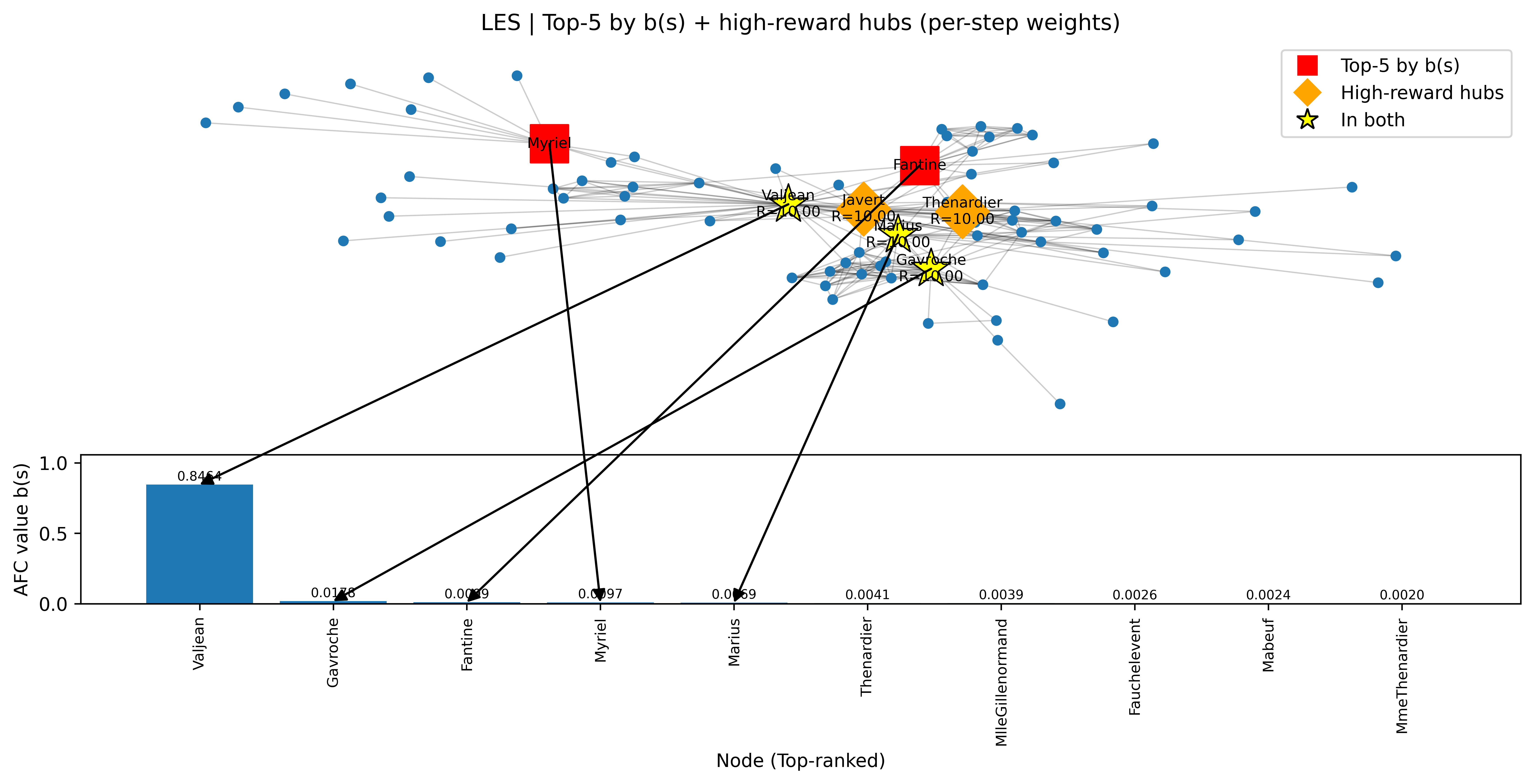}
      \includegraphics[width=0.49\linewidth]{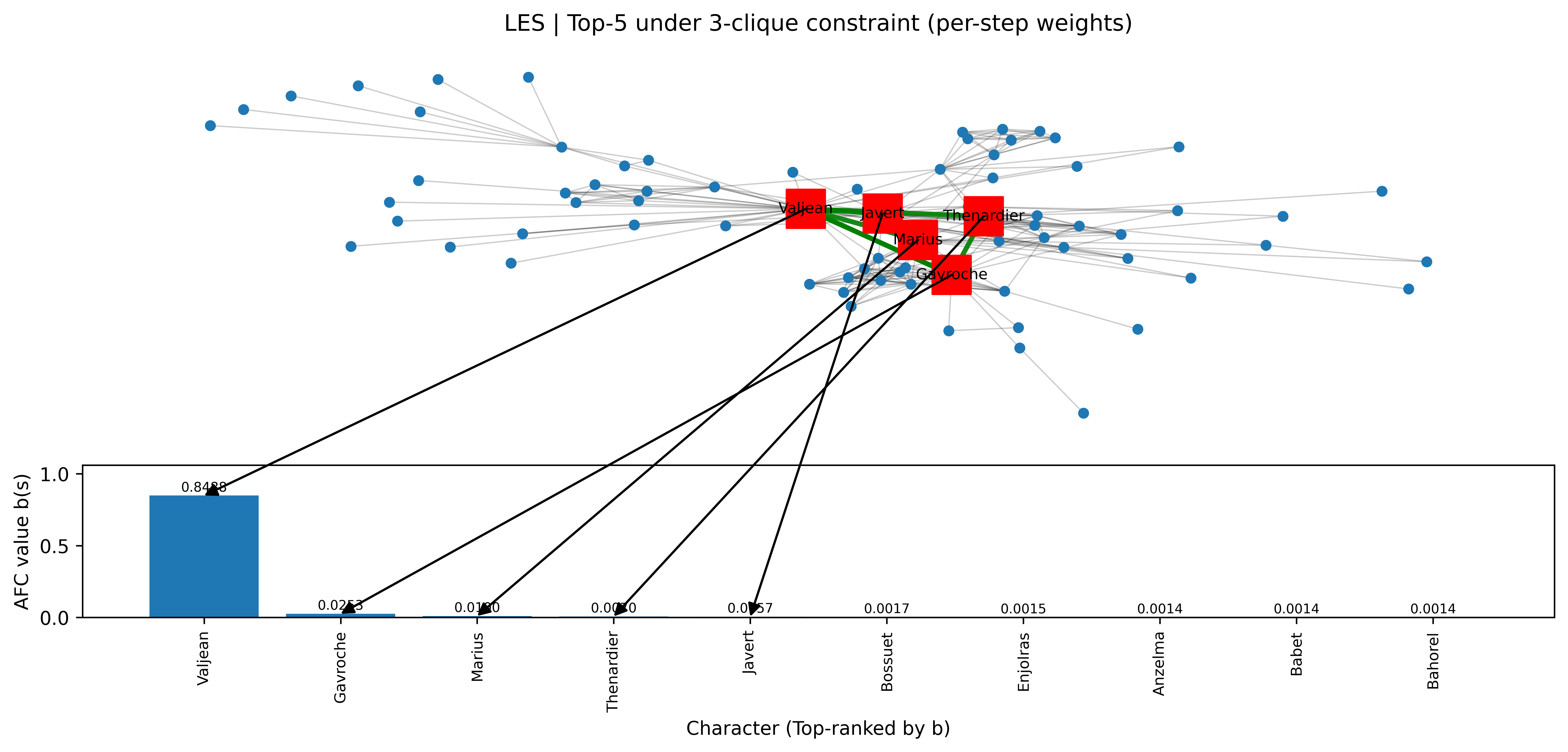}
    \caption{The upper-left panel reports the baseline AFC results. The upper-right panel compares the robust variants under KL divergence and $W_{1}$ discrepancy. The lower-left panel presents the multi-reward extensions. The lower-right panel shows the structure-constrained setting in which betweenness centrality selection is restricted to nodes contained in the chosen $3$-clique target pool.}
    \label{fig:test5}
\end{figure}

Figure~\ref{fig:test5} summarizes the Les Mis\'erables experiments with per-step edge-weight resampling. The upper-left panel shows the baseline occupancy profile by highlighting the Top-$5$ nodes and plotting the Top-$10$ entries of $b(s)$. The upper-right panel studies kernel robustness: starting from the Monte Carlo estimate $\widehat P^{0}$, we apply row-wise feasible perturbations to the transient block $Q$ and search for admissible kernels that maximize either KL divergence or $W_{1}$ distance from the baseline AFC. Unlike ER/WS (where variability is driven by random edge presence), the topology here is fixed and randomness arises from weight resampling, so the robust results reflect both admissible deviations in $Q$ and finite-sample uncertainty in $\widehat P^{0}$. For $W_{1}$, the ground metric is shortest-path distance on the base graph.

The lower-left panel reports the multi-reward test, combining the Top-$10$ AFC bar chart with degree-based reward hubs and reward-aware summaries (e.g., $b_f(s)=b(s)^\top f$ and transition-reward post-processing). The lower-right panel reports the structure-constrained test, where a target pool $W$ is formed from selected $3$-cliques and the one-step output is restricted to $\mathrm{Top}\text{-}k\cap W$ when feasible, with a deterministic fallback otherwise. Together, the four panels separate baseline AMC dynamics, kernel-level sensitivity, reward-aware post-processing, and clique-based constraints under a fixed topology with stochastic edge weights.

\section{Conclusion}

Our work provides a unified and computationally tractable absorbing dynamics framework for stochastic network centrality. Its main limitations are that the node-valued AMC compression is not lossless, since unresolved candidate sets and full local Top-$k$ outputs are reduced by deterministic tie-breaking, and that estimation currently relies on row-wise Monte Carlo with conservative finite-sample guarantees. Additionally, practical stabilization may also introduce bias, and better temporal dependence would require state augmentation, which could hurt its tractability.

These limitations suggest several directions for future work, including sharper non-asymptotic bounds, bias-aware analysis of stabilized estimators, uncertainty models for robust AFC optimization, and set-valued or state-augmented formulations that retain more information from local Top-$k$ structure and temporal dependence. On the applied side, the framework could be extended to problems such as resilient routing, infrastructure monitoring, uncertainty-aware prioritization, and dynamic screening in weighted social or co-occurrence networks.


\bibliography{references}

\appendix 

\section{AMC under an arbitrary initial distribution}\label{subsec:stochastic}

In this appendix, we discuss some initial results starting from an arbitrary initial distribution. In general, under component-wise centers, the next reported center depends on the current anchor $i$ through $c_{\mathrm{loc}}(i;H)$. Therefore the transient kernel $Q$ need not have identical rows.

This appendix presents a simpler connected/anchor-free regime. In that regime, conditional on survival, every post-initial reported center has the same distribution $p$, independent of the current anchor. The main consequence is that AFC becomes a simple mixture of the initial distribution $s$ and this common law $p$.

Let $\mathcal L$ be the law of one network state $Y$ (random topology and/or positive edge weights, possibly correlated within a step). Define the one-shot center distribution by
\begin{equation}\label{eq:oneshot-p}
p_v:=\mathbb P_{Y\sim\mathcal L}(c(Y)=v),\qquad p\in\Delta^{n-1}.
\end{equation}

Fix an initial distribution $s\in\Delta^{n-1}$. Let $(X_t)$ be the reported-center process on $V\cup\{\perp\}$ with absorption time $T$ as in \eqref{eq:absorption-time}, started from $X_0\sim s$. We continue to write $\pi_t(v):=\mathbb P_s(X_t=v\mid T>t)$, and $w_t:=\frac{\mathbb P_s(T>t)}{\mathbb E_s[T]}$, as in Corollary~\ref{cor:survival-decomp}. We assume:

\medskip
\noindent (S1) Stationary post-initial law conditional on survival. There exists $p\in\Delta^{n-1}$ such that $\pi_t=p$ for every $t\ge 1$. 

\medskip
\noindent (S2) Finite expected absorption time; that is, $T<\infty \text{ a.s.}, 1\le \mathbb E_s[T]<\infty.$

\begin{remark}[When does (S1) hold?]
A sufficient condition is that realized working graphs are almost surely connected, each realization has a unique center (after employing the tie-breaker), and for every $t\ge 1$ the random center $X_t=c(Y_t)$ has the same law $p$. Here \(p=(p_v)_{v\in V}\in\Delta^{n-1}\) denotes the common conditional distribution of the reported center after the initial step, that is, \(p_v=\mathbb P_s(X_t=v\mid T>t)\) for every \(t\ge 1\) and \(v\in V\). For example, this holds when $(Y_t)_{t\ge 1}$ are i.i.d.\ with law $\mathcal L$. In that case, $p$ is exactly the one-shot distribution in \eqref{eq:oneshot-p}.
\end{remark}

\begin{remark}[Disconnectedness of realized graphs]
\label{rem:disconnect-what-changes}
If realized graphs are disconnected, then the local center $c_{\mathrm{loc}}(i;H)=c_{\mathrm{comp}}(K(i;H);H)$ depends on the anchor through the random component containing $i$. A useful summary is the conditional continuation law
$$
p^{\mathrm{loc}}(i)_j
:=
\mathbb P(X_{t+1}=j\mid X_t=i,\ X_{t+1}\neq\perp)
=
\frac{Q_{ij}}{1-P_{i\perp}},
\qquad\text{whenever }P_{i\perp}<1.
$$
In general, $p^{\mathrm{loc}}(i)$ varies with $i$. Then (S1) may fail, and the mixture formula derived need not hold.
The general survival-weighted decomposition \eqref{eq:b-survival-average} still holds, and AFC is still computed from the
induced AMC through \eqref{eq:visits}--\eqref{eq:afc-def}.
\end{remark}

Under (S1), the general decomposition \eqref{eq:b-survival-average} simplifies because only the time-$0$ term depends on
the initial distribution.

\begin{theorem}[Mixture formula for AFC]
\label{thm:mixture}
Under (S1) and (S2), the absorbing-frequency centrality satisfies
\begin{equation}\label{eq:mixture-general}
b(s) = \frac{1}{\mathbb E_s[T]}\,s + \Bigl(1-\frac{1}{\mathbb E_s[T]}\Bigr)p.
\end{equation}
If, in addition, $\mathbb E_s[T]>1$ and we exclude the initial step by defining
\begin{equation}\label{eq:def-b-plus}
b^{(+)}(s)
:=
\dfrac{\mathbb E_s\!\big[\sum_{t=1}^{T-1}\mathbf 1\{X_t=\cdot\}\big]}{\mathbb E_s[T-1]},
\end{equation}
then $b^{(+)}(s)=p$ for every $s$.
\end{theorem}

\begin{proof}
By Corollary~\ref{cor:survival-decomp}, $b(s)=\sum_{t\ge 0} w_t\,\pi_t$. Under (S1), we have $\pi_0=s$ and $\pi_t=p$ for all $t\ge 1$. Since $T\ge 1$ almost surely, $w_0=\frac{\mathbb P_s(T>0)}{\mathbb E_s[T]}=\frac{1}{\mathbb E_s[T]}$, and $\sum_{t\ge 1} w_t = 1-w_0$. Therefore,
$
b(s)
=
w_0 s+\Bigl(\sum_{t\ge 1}w_t\Bigr)p
=
\frac{1}{\mathbb E_s[T]}s
+
\Bigl(1-\frac{1}{\mathbb E_s[T]}\Bigr)p,
$
which is the same as \eqref{eq:mixture-general}.

For the post-initial normalization, fix $v\in V$. Then
$
\mathbb E_s\!\Big[\sum_{t=1}^{T-1}\mathbf 1\{X_t=v\}\Big]
=
\sum_{t\ge 1}\mathbb P_s(T>t)\,\pi_t(v)
=
p_v \sum_{t\ge 1}\mathbb P_s(T>t).
$
Using
$
\mathbb E_s[T-1]=\sum_{t\ge 1}\mathbb P_s(T>t),
$
we obtain
$
\mathbb E_s\!\Big[\sum_{t=1}^{T-1}\mathbf 1\{X_t=v\}\Big]
=
p_v\,\mathbb E_s[T-1].
$
Dividing by $\mathbb E_s[T-1]>0$ gives $b^{(+)}(s)=p$.
\end{proof}

Theorem~\ref{thm:mixture} uses only the survival-conditional laws $(\pi_t)$; it does not require $(X_t)$ to be
time-homogeneous Markov. A particularly simple case is geometric stopping.

\begin{corollary}[Geometric stopping]
\label{cor:geometric}
Assume $T$ is geometric with parameter $\alpha\in(0,1)$ on $\{1,2,\dots\}$, so $\mathbb E_s[T]=\alpha^{-1}$ for every s. Under (S1), $b(s)=\alpha\,s+(1-\alpha)\,p, b^{(+)}(s)=p$. 
\end{corollary}

\begin{proof}
Substitute $\mathbb E_s[T]=\alpha^{-1}$ into \eqref{eq:mixture-general}.
\end{proof}

\begin{proposition}
\label{prop:linear-algebra}
Consider the \emph{canonical} AMC in which, from every transient state, the chain is absorbed with probability
$\alpha$, and otherwise survives and moves to a fresh draw from $p$, independently of the current node. Then $Q=(1-\alpha)\,\mathbf 1 p$, and $r=\alpha\,\mathbf 1.$ Its fundamental matrix is $N=(I-Q)^{-1}=I+\frac{1-\alpha}{\alpha}\,\mathbf 1 p$.
Consequently, $sN=s+\frac{1-\alpha}{\alpha}\,p$, and $sN\mathbf 1=\frac{1}{\alpha}$. Hence $b(s)=\alpha\,s+(1-\alpha)\,p$.
\end{proposition}

\begin{proof}
Under this construction, every transient row is the same: with probability $\alpha$ the chain is absorbed, and with probability $1-\alpha$ it survives and the next node is drawn from $p$. This gives $Q=(1-\alpha)\,\mathbf 1 p$, and $r=\alpha\,\mathbf 1.$

Since $p\mathbf 1=1$, we have $(\mathbf 1 p)^2=\mathbf 1(p\mathbf 1)p=\mathbf 1 p$, so for every $k\ge 1$, $Q^k=(1-\alpha)^k\,\mathbf 1 p$. Therefore
$
N=\sum_{k\ge 0}Q^k
=
I+\sum_{k\ge 1}(1-\alpha)^k\,\mathbf 1 p
=
I+\frac{1-\alpha}{\alpha}\,\mathbf 1 p.
$

Multiplying by $s$ gives $sN=s+\frac{1-\alpha}{\alpha}\,p$, and since $p\mathbf 1=1$, $sN\mathbf 1
=
s\mathbf 1+\frac{1-\alpha}{\alpha}\,p\mathbf 1
=
1+\frac{1-\alpha}{\alpha}
=
\frac{1}{\alpha}.
$
Finally,
$
b(s)=\frac{sN}{sN\mathbf 1}=\alpha\,s+(1-\alpha)\,p.
$
\end{proof}

\begin{remark}[When the mixture formula fails]
If the survival-conditional laws $\pi_t$ vary with $t$, or if the stopping mechanism changes the law of $X_t$ given
survival, then \eqref{eq:mixture-general} need not hold. In that case one should use the general decomposition
\eqref{eq:b-survival-average} and the uniform pre-absorption-step interpretation from
Proposition~\ref{prop:uniform-step}.
\end{remark}

\section{A worked illustration for section \ref{sec:rankreversals}}
\label{appendixB}
Here we show a small, worked illustration for section \ref{sec:rankreversals}. Take $V=\{1,2,3\}$, uniform $s$, and nominal transient block and absorption probabilities
\[
Q^0=\begin{bmatrix}0.20&0.10&0\\[1pt]0.05&0.20&0\\[1pt]0.05&0.05&0.10\end{bmatrix},
\qquad r^0=(0.70,\,0.75,\,0.80),
\]
giving $\mu(s;P^0)=(0.4710,\,0.4987,\,0.3704)$, $\mathbb E_s[T]=1.34$, and visit-gaps $G_{21}=0.0277$, $G_{23}=0.1283$, $G_{13}=0.1006$. With entry-wise radii $\varepsilon_{ij}=0.004$ (so $\bar\varepsilon=0.012$) and $\underline r_{\min}=0.65$, the crude threshold is $2\bar\varepsilon/\underline r_{\min}^{\,2}=0.0568$: it certifies $2\succ3$ and $1\succ3$ but is inconclusive for the boundary pair $2\succ1$. The sharpened thresholds from \eqref{eq:sharp-gap} are $0.0218$, $0.0233$, $0.0233$ for the pairs $(2,1)$, $(2,3)$, $(1,3)$ respectively, so Proposition~\ref{prop:sharp-certificate} certifies all three orders, including the boundary pair that the crude bound could not resolve.

Practically, we note two caveats. First, certificates of this kind are informative only when the radii are calibrated to the actual uncertainty: at the radii used in the experiments of Subsection~\ref{subsec:exp2-robust} ($\delta_{\mathrm{rel}}=0.50$ multiplicative, $r_{\min}=0.05$), the row radius is $\bar\varepsilon\approx 0.5(1-\alpha)\approx 0.43$, the crude threshold equals $2\bar\varepsilon/r_{\min}^{2}\approx 340$, and every visit-gap is bounded by $\mathbb E_s[T]\le 1/r_{\min}$, so no pair can be certified, which simply reflects that $50\%$ relative perturbations genuinely reorder near-ties. Second, when the kernel is estimated by Algorithm~\ref{alg:construct-amc}, the natural radii are the finite-sample deviations $\varepsilon_{ij}\asymp\sqrt{\log(2n(n+1)/\delta)/(2M)}$ from Subsection~\ref{subsec:sample-size}; the revised experiments report, for each boundary pair of the nominal Top-$k$ set, whether \eqref{eq:sharp-gap} certifies the order at these statistically calibrated radii.

\end{document}